\documentclass{article}
\usepackage{ijcai19}

\usepackage{times}
\usepackage{soul}
\usepackage{url}
\usepackage[hidelinks]{hyperref}
\usepackage[utf8]{inputenc}
\usepackage[small]{caption}
\usepackage{subcaption}
\usepackage{sidecap}
\usepackage{graphicx}
\usepackage{amsmath}
\usepackage{booktabs}
\usepackage{amssymb}

\usepackage{algorithm}
\usepackage{algorithmicx}
\usepackage[noend]{algpseudocode}

\DeclareMathOperator*{\E}{\mathbb{E}}
\newcommand{\polym}{\mathcal{P}}
\newcommand{\norm}[1]{\left\lVert#1\right\rVert}
\newcommand{\oracle}[2]{\mathcal{A}^{#1}_{\text{#2}}}

\usepackage[utf8]{inputenc}
\usepackage{mathtools}
\usepackage{titlecaps} 
\usepackage{amsthm}

\usepackage{colortbl}
\usepackage{xcolor}

\newtheorem{theorem}{Theorem}[section]

\newtheorem{lemma}[theorem]{Lemma}
\newtheorem{conjecture}[theorem]{Conjecture}

\newcommand{\Rational}{\mathrm{Rational}}
\newcommand{\OPT}{\mathrm{OPT}}
\newcommand{\Maximin}{\mathrm{Maximin}}

\newcommand{\Infl}{\mathcal{I}}
\newcommand{\PoF}{PoF}
\DeclareMathOperator*{\argmax}{argmax} 
\newcommand{\euler}{\mathrm{e}}

\DeclarePairedDelimiter{\ceil}{\lceil}{\rceil}

\newcommand\RegTextCaption[1]{%
  \captionsetup{font=normalsize}%
  \caption*{#1}}

\newcommand{\citet}[1]{\citeauthor{#1} \shortcite{#1}}

\title{Group-Fairness in Influence Maximization}

\author{
Alan Tsang$^{1*}$
\and
Bryan Wilder$^{2*}$\and
Eric Rice$^{2}$\And
Milind Tambe$^2$\and
Yair Zick$^1$
\affiliations
$^*$Equal contribution\\
$^1$Department of Computer Science, National University of Singapore\\
$^2$Center for AI in Society, University of Southern California\\
\emails
akhtsang@nus.edu.sg,
\{bwilder, ericr, tambe\}@usc.edu, 
zick@comp.nus.edu.sg
}

\begin{document}
\maketitle

\begin{abstract}
Influence maximization is a widely used model for information dissemination in social networks.
Recent work has employed such interventions across a wide range of social problems, spanning public health, substance abuse, and international development (to name a few examples). A critical but understudied question is whether the benefits of such interventions are {\em fairly distributed} across different groups in the population; e.g., avoiding discrimination with respect to sensitive attributes such as race or gender. Drawing on legal and game-theoretic concepts, we introduce formal definitions of fairness in influence maximization. 
We provide an algorithmic framework to find solutions which satisfy fairness constraints, and in the process improve the state of the art for general multi-objective submodular maximization problems. Experimental results on real data from an HIV prevention intervention for homeless youth show that standard influence maximization techniques oftentimes neglect smaller groups which contribute less to overall utility, resulting in a disparity which our proposed algorithms substantially reduce. 
\end{abstract}

\section{Introduction}
Influence maximization in social networks is a well-studied problem with applications in a broad range of domains. Consider, for example, a group of at-risk youth; outreach programs try to provide as many people as possible with useful information (e.g., HIV safety,  or available health services). Since resources (e.g., social workers) are limited, it is not possible to personally reach every at-risk individual. It is thus important to target {\em key community figures} who are likely to spread vital information to others. Formally, individuals are nodes $V$ in a social network, and we would like to influence or \emph{activate} as many of them as possible. This can be done by initially {\em seeding} $k$ nodes (where $k\ll |V|$). The seed nodes activate their neighbors with some probability, who activate their neighbors and so forth. Our goal is to identify $k$ seeds such that the maximal number of nodes is activated. This is the classic {\em influence maximization problem} \cite{kempe2003maximizing}, that has received much attention in the literature. 

In recent years, the influence maximization framework has seen application to many social problems, such as HIV prevention for homeless youth \cite{yadav2018bridging,wilder2018end}, public health awareness \cite{valente2007identifying}, financial inclusion \cite{banerjee2013diffusion}, and more. Frequently, small and marginalized groups within a larger community are those who benefit the most from attention and assistance. It is important, then, to ensure that the allocation of resources reflects and respects the diverse composition of our communities, and that each group receives a fair allocation of the community's resources. For instance, in the HIV prevention domain we may wish to ensure that members of racial minorities or of LGBTQ identity are not disproportionately excluded; this is where our work comes in. 

\paragraph{Our Contributions: }
This paper introduces the problem of fair resource allocation in influence maximization. Our \emph{first contribution} is to propose fairness concepts for influence maximization. We start with a \emph{maximin} concept inspired by the legal notion of disparate impact; formally it requires us to maximize the minimum fraction of nodes within each group that are influenced. While intuitive and well-motivated, this definition suffers from shortcomings that lead us to introduce a second concept, \emph{diversity constraints}. Roughly, diversity constraints guarantee that every group receives influence commensurate with its ``demand'', i.e., what it could have generated on its own, based on a number of seeds proportional to its size. Here, to compute a group's demand, we allow it a number of seeds proportional to its size, but require that it spreads influence using only nodes in the group. Hence, a small but well connected group may have a better claim for influence than a large but sparsely connected group. 

Our \emph{second contribution} is an algorithmic framework for finding solutions that satisfy either fairness concept. While the classical influence maximization problem is submodular (and hence easily solved with a greedy algorithm), fairness considerations produce strongly non-submodular objectives. This renders standard techniques inapplicable. We show that both fairness concepts can be reduced to \emph{multi-objective} submodular optimization problems, which are substantially more complex. Our key algorithmic contribution is a new method for general multi-objective submodular optimization which has substantially better approximation guarantee than the current best algorithm \cite{udwani2018multi}, and often better runtime as well. This result may be of independent interest. 

Our \emph{third contribution} is an analytical exploration of the \emph{price of group fairness} in influence maximization, i.e., the reduction in social welfare with respect to the unconstrained influence maximization problem due to imposing a fairness concept. We show that the price of diversity can be high in general for both concepts and under a range of settings. 

Our \emph{fourth contribution} is an empirical study on real-world social networks that have been used for a socially critical application: HIV prevention for homeless youth. Our results show that standard influence maximization techniques often cause substantial fairness violations by neglecting small groups. Our proposed algorithm substantially reduces such violations at relatively small cost to overall utility.

\paragraph{Related Work: }
\citet{kempe2003maximizing} introduced influence maximization and proved that since the objective is submodular, greedily selecting nodes gives a $\left(1 - \frac1\euler\right)$-optimal solution. There has since been substantial interest among the AI community both in developing more scalable algorithms (see \citet{li2018influence} for a recent survey) 
, as well as in addressing the challenges of deployment in public health settings \cite{yadav2016using,wilder122018maximizing}. Recently, such algorithms have been used in real-world pilot tests for HIV prevention amongst homeless youth \cite{yadav2018bridging,wilder2018end}, driving home the  need to consider fairness as influence maximization is applied in socially sensitive domains. To our knowledge, no previous work considers fairness specifically for influence maximization. The techniques we introduce to optimize fairness metrics are related to research on multi-objective submodular maximization (outside the context of fairness), and we improve existing theoretical guarantees for this general problem \cite{chekuri2010dependent,udwani2018multi}. 

Outside of influence maximization, the general idea of diversity as an optimization constraint has received considerable attention in recent years; it has been studied in multiwinner elections (see ~\cite{bredereck2018diversity,faliszewski2018multi} for an overview), resource allocation \cite{benabbou2018diversity,aghaei2019learning}, and matching problems \cite{ahmed2017diversity,hamada2017quotas}. 
We note that some of the above works (e.g. \citet{ahmed2017diversity} and \citet{shcumann2017diverse}) use a submodular objective function as a means of achieving diversity; interestingly, while the classic influence maximization target function is submodular, it is no longer so under diversity constraints.
Group fairness has been studied extensively in the voting theory literature, where the objective is to identify a committee of $k$ candidates that will satisfy subsets of voters (see a comprehensive overview in \cite{faliszewski2018multi}). 
There have also been several works on group fairness in fair division, defining notions of group envy-freeness \cite{conitzer2019group,fain2018fair,segal2018fair,todo2011group}, and a group maximin share guarantee \cite{barman2019group,suksompong2018maximin}.   
\section{Model}
Agents are embedded in a social network $G = (V,E)$.  An edge $(i,j) \in E$ represents the ability for agent $v_i$ to influence or \emph{activate} $v_j$. $G$ may be undirected or directed.

\paragraph{Diversity: } Each agent in our network may identify with one or more groups within the larger population.  These represent different ethnicities, genders, sexual orientations, or other groups for which fair treatment is important. Our goal is to maximize influence in a way such that each group receives at least a ``fair'' share of influence (more on this below). Let us designate these groups as $\mathcal{C} = \{C_1, \ldots C_m\}$.  Each {\bf group} $C_i$ represents a non-empty subset of V, $\emptyset \neq C_i \subseteq V$.  Each agent must belong to at least one group, but may belong to multiple groups; i.e. $C_1 \cup C_2 \cup \ldots C_m = V$.  In particular, this allows for the expression of intersectionality, where an individual may be part of several minority groups.

\paragraph{Influence maximization: }

We model influence using the \emph{independent cascade model} \cite{kempe2003maximizing}, the most common model in the literature. All nodes begin in the inactive state. The decision maker then selects $k$ \emph{seed nodes} to activate. Each node that is activated makes one attempt to activate each of its inactive neighbors; each attempt succeeds independently with probability $p$. Newly activated nodes attempt to activate their neighbors and so on, with the process terminating once there are no new activations.

We define the \textbf{influence} of nodes $A \subseteq V$, denoted $\Infl_G(A)$, as the expected number of nodes activated by seeding $A$. Of these, let $\Infl_{G,C_i}(A)$ be the expected number of activated vertices from $C_i$. Traditional influence maximization seeks a set $A$, $|A| \leq k$, maximizing $\Infl_G(A)$. Using a slight abuse of notation, let $\Infl_G(k)$ be the maximum influence that can be achieved by selected $k$ seed nodes.  That is, $\Infl_G(k) = \max_{|A|=k} \Infl_G(A)$.  Analogously, we define $\Infl_{G,C_i}(k)$ as the maximum expected number of vertices from $C_i$ that can be activated by $k$ seeds. We now propose two means of capturing group fairness in influence maximization.

\paragraph{Maximin Fairness: }

Maximin Fairness captures the straightforward goal of improving the conditions for the \textit{least well-off} groups.  That is, we want to maximize the minimum influence received by any of the groups, as proportional to their population.  This leads to the following utility function:

\[
U^\Maximin(A) = \min_{i} \frac{\Infl_{G,C_i}(A)}{|C_i|}
\]

Subject to this maximin constraint, we seek to maximize overall influence.
Thus, we define  $\Infl_G^\Maximin = \Infl_G(B)$ with $B = \argmax_{A \subseteq V, |A|=k} U^\Maximin(A)$.  That is, $\Infl_G^\Maximin$ is the expected number of nodes activated by a seed configuration that maximizes the minimum proportional influence received by any group. This corresponds to the legal concept of \emph{disparate impact}, which roughly states that a group has been unfairly treated if their ``success rate" under a policy is substantially worse than other groups (see \cite{barocas2016big} for an overview). Therefore, maximin fairness may be significant to governmental or community organizations which are constrained to avoid this form of disparity. However, optimizing for equality of outcomes may be undesirable when some groups are simply much better suited than others to a network intervention. For instance, if one group is very poorly connected, maximin fairness would require that large number of nodes be spent trying to reach this group, even though additional seeds have relatively small impact.





\paragraph{Diversity Constraints: }

We now propose an alternate fairness concept by extending the notion of individual rationality to \textbf{Group Rationality}. The key idea is that no group should be better off by leaving the (influence maximization) game with their proportional allocation of resources and allocating them internally. For each group $C_i$, let $k_i = \ceil{ k |C_i| / |V|}$ be the number of seeds that would be fairly allocated to the group $C_i$ based on the group's size within the larger population, rounded up to remove any doubt that this group receives a fair share.
$k_i$ is the \textbf{fair allocation} of seeds to the group.  

Let $G[C_i]$ be the subgraph induced from $G$ by the nodes $C_i$.  This represents the network formed by group $C_i$ if they were to separate from the original network. 
Now, we define the \textbf{group rational influence} that each group $C_i$ can expect to receive as the number of nodes they expect to activate if they left the network, with their fair allocation of $k_i$ seeds.  We denote this group rational influence for $C_i$ as $\Infl_{G[C_i]}(k_i)$. Then, we devise a set of \textbf{diversity constraints} that any group rational seeding configuration $A$ with $k$ seeds must satisfy: $\Infl_{G,C_i}(A) \geq \Infl_{G[C_i]}(k_i), \forall i$.  That is, the influence received by each group is at least equal to what each group may accomplish on its own when given its fair share of $k_i$ seed nodes.

The diversity constraint objective function is to maximize the expected number of nodes activated, subject to the above diversity constraint.  The utility for selecting seed nodes $A$ is:

\begin{align*}
U^\Rational(A) =~
\begin{cases}
	\Infl_G(A), &\text{if}~\Infl_{G,C_i}(A) \geq \Infl_{G[C_i]}(k_i),
		\forall i. \\
	0, &\text{otherwise}.
\end{cases}	
\end{align*}

The maximum expected influence obtained via a group rational seeding configuration $A$ is called the \textbf{rational influence} $\Infl_{G}^\Rational = \Infl_{G}(B)$, where $B = \argmax_{A \subseteq V,|A|=k} U^\Rational(A) $.

\paragraph{Price of Fairness: }
To measure the cost of ensuring a fair outcome for the diverse population, we will measure the Price of Fairness, the ratio of optimal influence to the best achievable influence under our two fairness criteria. Here \textbf{optimal influence} $\Infl_{G}^\OPT = \Infl_{G}(k)$, which is the maximum amount of expected influence that can be obtained using any choice of $k$ seed nodes.  We omit the subscript where the context is clear. 

\[
\PoF^\Rational = \frac{\Infl^\OPT}{\Infl^\Rational} \hspace{0.25cm}
\PoF^\Maximin = \frac{\Infl^\OPT}{\Infl^\Maximin}
\]

\section{Optimization}


The standard approach to influence maximization is based on \emph{submodularity}. Formally, a set function $f$ on ground set $V$ is submodular if for every $A \subseteq B \subseteq V$ and $x \in V \setminus B$, $f(A \cup \{x\}) - f(A) \geq f(B \cup \{x\}) - f(B)$. This captures the intuition that additional seeds provide diminishing returns. However, both of our fairness concepts are easily shown to violate this property (proofs are deferred to the appendix):

\begin{theorem} \label{thm:not-submod}
$U^\Maximin$ and $U^\Rational$ are not submodular.
\end{theorem}





Hence, we cannot apply the greedy heuristic to group-fair influence maximization.  However, we now show that optimizing either utility function reduces to \emph{multiobjective submodular maximization}, for which we we give an improved algorithm below. Consider the following generic problem: given monotone submodular functions $f_1...f_m$ and corresponding target values $W_1...W_m$, find a set $S$ satisfying $|S| \leq k$ with $f_i(S) \geq W_i$ for all $i$ (under the promise that such an $S$ exists). Roughly, $f_i$ will be group $i$'s utility, and $W_i$ will be the utility that we want to guarantee for $i$.  Suppose that we have an algorithm for the above multiobjective problem. Then, we can optimize the maximin objective by letting $f_i = \frac{\Infl_{G, C_i}}{|C_i|}$ and binary searching for the largest $W$ such that $f_i \geq W$ is feasible for all groups $i$. For diversity constraints, we let $f_i = \Infl_{G, C_i}$ and set the target $W_i = \Infl_{G[C_i]}(k_i)$. We then add another objective function $f_{\text{total}} = \Infl_{G}$ representing the combined utility and binary search for the highest value $W_{\text{total}}$ such that the targets $W_1...W_m, W_{\text{total}}$ are feasible. This represents the largest achievable total utility, subject to diversity constraints. Having reduced both fairness concepts to multiobjective submodular maximization, we now give an improved algorithm for this core problem.



The multiobjective submodular problem was introduced by Chekuri et al.\ \shortcite{chekuri2010dependent}, who gave an algorithm which guarantees $f_i \geq (1 - \frac{1}{e})W_i$ for all $i$ provided that the number of objectives $m$ is smaller than the budget $k$ (when $m = \Omega(k)$, the problem is provably inapproximable \cite{krause2008robust}). Unfortunately, this algorithm is of mostly theoretical interest since it runs in time $O(n^8)$. Udwani \shortcite{udwani2018multi} recently introduced a practically efficient algorithm; however it obtains an asymptotic $(1 - \frac{1}{e})^2$-approximation instead of the optimal $\left(1 - \frac{1}{e}\right)$. We remedy this gap by providing a practical algorithm obtaining an asymptotic $\left(1 - \frac{1}{e}\right)$-approximation (Algorithm \ref{alg:multi-submodular}). Its runtime is comparable to, and under many conditions faster than, the algorithm of \cite{udwani2018multi}. We present the high-level idea behind the algorithm here, with additional details present in the appendix.

Previous algorithms \cite{chekuri2010dependent,udwani2018multi} start from a common template in submodular optimization, which we also build on. The main idea is to relax the discrete problem to a continuous space. For a given submodular function $f$, its \emph{multilinear extension} $F$ is defined on $n$-dimensional vectors $x$ where $0 \leq x_j \leq 1$ for all $j$. $x_j$ represents the probability that item $j$ is included in the set. Formally, let $S \sim x$ denote a set which includes each $j$ independently with probability $x_j$. Then, we define $F(x) = \E_{S \sim x}[f(S)]$, which can be evaluated using random samples. 

The main challenge is to solve the continuous optimization problem, which is where our technical contribution lies. Algorithm \ref{alg:multi-submodular} describes the high-level procedure, which runs our continuous optimization subroutine (line 2) and then rounds the output to a discrete set (line 3). Line 1, which ensures that all items with value above a threshold $\tau$ are included in the solution, is a technical detail needed to ensure the rounding succeeds. The rounding process captured in lines 1 and 3 is fairly standard and used by both previous algorithms \cite{chekuri2010dependent,udwani2018multi}. Our main novelty lies in an improved algorithm for the continuous problem, \textsc{MultiFW}. 

	\begin{algorithm}
		\caption{Multiobjective Optimization$(\gamma, \tau, T, T', 
		\eta)$}\label{alg:multi-submodular}
		\begin{algorithmic}[1] 
			\State $S_1 =  \{j : f_i(\{j\}) \geq \tau$ for some $i\}$
			\State $x =$\textsc{MultiFW}$(k-|S_1|, \{\gamma\left(W_i - f_i(S)\right)\}_{i=1}^m)$
			\State $S_2 = $\textsc{SwapRound}$(x_{\text{int}})$ //see \cite{chekuri2010dependent}
			\State\Return $S_1 \cup S_2$
		\end{algorithmic}
	\end{algorithm}

	\begin{algorithm}
		\caption{Multiobjective Frank-Wolfe$(k, \{W_i\})$}\label{alg:multifw}
		\begin{algorithmic}[1] 
			\State $x^0 = 0$
			\For{$t=1...T$}
			\State $v^t = \textsc{S-SP-MD}(x, \{i: W_i - F_i(x^{t-1}) \geq \epsilon\})$
			\State $x^{t} = x^{t-1} + \frac{1}{T}v^t$
			\EndFor
			\State\Return $ \textsc{ApproxDecomposition}(x^T)$ //see \cite{mirrokni2017tight}
			
			\Function{S-SP-MD}{$x, \mathcal{I}$}
			\State Initialize $v$ s.t. $||v||_1  = k$ and $y \in \Delta(\mathcal{I})$ arbitrarily
			\For{$\ell = 1...T'$}
			\State Sample $i \sim y$; set $\hat{\nabla}_v = \frac{1}{W_i - F_i(x)}\oracle{i}{grad}(x)$
			\State Sample $j \sim v$;  $\hat{\nabla}_y = k \cdot  \text{diag}\left(\frac{1}{\vec{W} - \vec{F}(x)}\right)\oracle{j}{item}(x)$
			\State $y = \frac{y e^{-\eta \hat{\nabla}_y}}{||y e^{-\eta \hat{\nabla}_y}||_1}$
			\State $v = k \frac{\min\{v e^{\eta \hat{\nabla}}_v, 1\}}{||\min\{v e^{\eta \hat{\nabla}}_v, 1\}||_1}$
			\EndFor
			\EndFunction
		\end{algorithmic}
	\end{algorithm}


\textsc{MultiFW} implements a Frank-Wolfe style algorithm to simultaneously optimize the multilinear extensions $F_1...F_m$ of the discrete objectives. The algorithm proceeds over $T$ iterations. Each iteration first identifies $v^t$, a good feasible point in continuous space (Algorithm \ref{alg:multifw}, line 3). Then, the current solution $x^t$ is updated to add $\frac{1}{T}v^t$ (line 4). The final output is an approximate decomposition of $x^T$ into integral points, produced using the algorithm of \cite{mirrokni2017tight}. This is a technical step required for the rounding procedure. 

The key challenge is to efficiently find a $v^t$ that makes sufficient progress towards \emph{every} objective simultaneously. We accomplish this by introducing the subroutine \textsc{S-SP-MD} (lines 6-12), which runs a carefully constructed version of stochastic saddle-point mirror descent \cite{nemirovski2009robust}. The idea is to find a $v$ for which $v \cdot \nabla_i F_i(x^{t-1})$ is large enough for all objectives $i$. We convert this into the saddle point problem of maximizing $\min_{i \in \mathcal{I}} v \cdot \nabla_i F_i(x^{t-1})$. $\mathcal{I}$ denotes the set of objectives $i$ where $W_i  - F_i(x^{t-1}) \geq \epsilon$ (i.e., those where we still need to make progress). We let $\Delta(\mathcal{I})$ denote the set of all distributions over $i$. Our approach only requires \emph{stochastic} gradients, a necessary feature since computing $\nabla_i F(x^{t-1})$ exactly may be intractable when the objective itself is randomized (as in influence maximization). 

Specifically, we assume access to two gradient oracles.  First, a stochastic gradient oracle $\mathcal{A}^i_{\text{grad}}$ for each multilinear extension $F_i$. Given a point $x$, $\mathcal{A}^i_{\text{grad}}(x)$ satisfies $\E[\mathcal{A}^i_{\text{grad}}] = \nabla_x F_i(x)$. Second, a stochastic gradient oracle $\mathcal{A}^j_{\text{item}}$ corresponding to each item $j \in [n]$ (in influence maximization, the items are the potential seed nodes). $\mathcal{A}^j_{\text{item}}(x)$ satisfies $\E[\mathcal{A}^j_{\text{item}}(x)] = \left[\nabla_{x_j} F_1(x) ... \nabla_{x_j} F_m(x)\right]$. We assume that $||\mathcal{A}^i_{\text{grad}}(x)||_\infty, ||\mathcal{A}^j_{\text{item}}(x)||_\infty \leq c$ for some constant $c$. Linear-time oracles are available for many common submodular maximization problems (e.g., coverage functions and facility location \cite{karimi2017stochastic}). Given such oracles, we implement a stochastic mirror descent algorithm for the maximin problem. We can interpret the algorithm as solving a game between the max player and the min player. The max player controls $v$, while the min player controls a variable $y$ representing the weight put on each objective. Intuitively, the min player will put large weights where the max player is doing badly, forcing the max player to improve $v$. Formally, in each iteration, the players take exponentiated gradient updates (lines 8-12). The max player obtains a stochastic gradient by sampling an objective with probability proportional to the current weights $y$, while the min player samples an item proportional to $v$ and uses that item's contribution to estimate the max player's current performance on each objective. We prove that these updates converge rapidly to the optimal $v$. With the subroutine in hand, our main algorithmic result is the following guarantee for Algorithm \ref{alg:multi-submodular}.  Here, $b = \max_{i, j} f_i(\{j\})$ is the maximum value of a single item.

	\begin{theorem}
		Given a feasible set of target values $W_1...W_n$, Algorithm \ref{alg:multi-submodular} outputs a set $S$ such that $f_i(S) \geq (1-\epsilon)\left(1 - \frac{m}{k(1 +\epsilon')\epsilon^3}\right)\left(1 - \frac{1}{e}\right)W_i - \epsilon$ with probability at least $1-\delta$. Asymptotically as $k\to\infty$, the approximation ratio can be set to approach $1 - 1/e$ so long as $m= o(k \log^3 k)$. The algorithm requires $O(nm)$ $\epsilon'$-accurate value oracle calls, $O(m\frac{bk^2}{\epsilon}\log \frac{1}{\delta})$ $\epsilon$-accurate value oracle calls, $O\left(\frac{bk^4 c^2}{\epsilon^5}  \log\left(n+ \frac{bk}{\delta\epsilon}\right)\right)$ calls to $\oracle{}{grad}$ and $\oracle{}{item}$, and $O\left(\frac{nk^2b^2}{\epsilon^2} + \frac{mk^2b}{\epsilon} + \frac{k^3 b^2}{\epsilon^2}\right)$ additional work. 
	\end{theorem}
	
	This says that Algorithm \ref{alg:multi-submodular} asymptotically converges to a $\left(1 - \frac{1}{e}\right)$-approximation when the budget $k$ is larger than the number of objectives $m$ (i.e., the conditions under which the problem is approximable). All terms in the approximation ratio are identical to Udwani \shortcite{udwani2018multi}, except that we improve their factor $\left(1 - \frac{1}{e}\right)^2$ to $\left(1 - \frac{1}{e}\right)$. The runtime is also identical apart from the time to solve the continuous problem (\textsc{MultiFW} vs their corresponding subroutine). This is difficult to compare since our respective algorithms use different oracles to access the functions. However, both kinds of oracles can typically be (approximately) implemented in time $O(n)$. Udwani's algorithm uses $O(n)$ oracle calls, while our's requires $O(bk^4c^2 \log n)$. For large-scale problems, $n$ typically grows much faster than $k$, $b$, and $c$ (all of which are often constants, or near-so). Hence, trading $O(n^2)$ runtime for $O(n\log n)$ can represent a substantial improvement. We present a more detailed discussion in the appendix.

	To instantiate Algorithm \ref{alg:multi-submodular} for influence maximization, we just need to supply appropriate stochastic gradient oracles. To our knowledge, no such oracles were previously known for influence maximization, which is substantially more complicated than other submodular problems because of additional randomness in the objective; naive extensions of previous methods require $O(n^2)$ time. We provide efficient $O(k n \log n)$ time stochastic gradient oracles by introducing a randomized method to simultaneously estimate many entries of the gradient at once (details may be found in the appendix).

\section{Price of Fairness}

In this section, we show that both definitions for the Price of Fairness can be unbounded; moreover, allowing nodes to join multiple groups can, counter-intuitively, worsen the PoF.
The proofs in this section use undirected graphs.  As they are more restrictive, the result naturally hold for directed graphs.

\begin{theorem}
As $n \to \infty$ and $p \to 0$, $\PoF^\Rational \to \infty$.
\end{theorem}

\begin{proof}

We construct a graph $G$ with two parts.  In Part $L$, we have $s-1$ vertices all disjoint except for two vertices; label one of these $x3$.  In Part $S$, we have a star with $s+1$ nodes.  Label a leaf node $x_1$ and the central node $x_2$.  We define two groups: $C_1$ is comprised of the $s$ degree-1 vertices of $S$, and $C_2$ for the remaining $s$ vertices, which includes the vertices of $L$ and the central vertex $x_2$ of the star.  There are $k=2$ seeds, and since $|C_1|=|C_2|$, they each have a fair allocation of $k_1=k_2=1$ seeds.  The figure below illustrates this network.

\begin{SCfigure}[2][h]
\centering
\includegraphics[width=0.3\columnwidth]{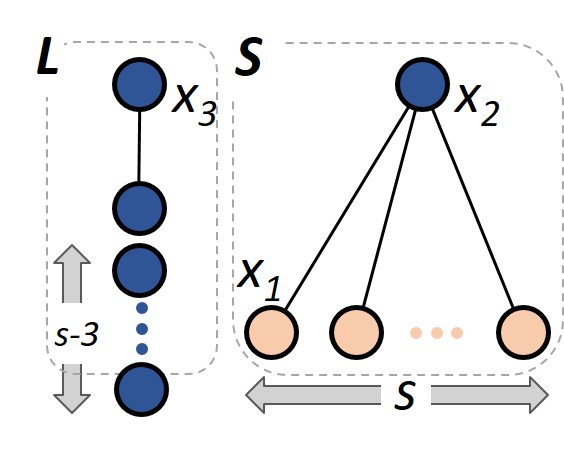}
\RegTextCaption{Since the subgraph induced by $C_1$ is comprised of isolated vertices, they have a rational allocation of $\Infl_{G[C_1]}(1) = 1$.  The subgraph induced by $C_2$ is a collection of isolated vertices and a $K_2$, its rational allocation is $\Infl_{G[C_2]}(1) = 1+p$.}
\label{fig:undirected-example}
\end{SCfigure}

We are interested in two seeding configurations: $A = \{x_1,x_3\}$ and $B= \{x_2,x_3\}$.  We can verify that configuration $A$ is fair.  The $A$ activates $1+p$ nodes in Part $L$, and $1+p+(s-1)p^2$ in Part $S$, for a total of $\Infl_{G}(A) = 2+2p+(s-1)p^2$.

Now consider configuration $B$.  $C_1$ receives $ps$ influence, and since $p < \frac{2}{n} = \frac{1}{s}$, $C_1$ does not receive its group rational share of influence.  However, we can verify that this seeding is optimal.  Part $L$ receives $(1+p)$ influence, and Part $S$ receives $1+ps$.  Therefore, $\Infl_{G}(B) = 2+p+ps$.

We may then calculate our Price of Fairness:

\[
\PoF^\Rational = \frac{\Infl_{G}^\OPT}{\Infl_{G}^\Rational} 
     = \frac{2+p+ps}{2+2p+(s-1)p^2}
\]

And if we take the limit as $n \to \infty$, $s \to \infty$, $\PoF \to 1/p$.  Finally, as as $p \to 0$, $\PoF \to \infty$.
\end{proof}

\noindent The appendix details a similar result for Maximin Fairness:

\begin{theorem} \label{thm:maximin_pof}
$\PoF^\Maximin$ is unbounded.
\end{theorem}

Frequently, an individual may identify with multiple groups.  Intuitively, we might expect such multi-group membership to improve the influence received by different groups and make the group-fairness easier to achieve (see the appendix for an example).  However, in this section, we show that this is not always true, and giving even a single node membership in a second group can cause the Price of Fairness to worsen by an arbitrarily large amount.

\begin{theorem}
Given graphs $G$ with groups $C_1$ and $C_2$, and $G'$ with groups $C'_1$ and $C'_2$, where $G'=G$, $C'_1 = C_1$ and $C'_2$ is obtained from $C_2$ by the addition of one vertex $x_1$ ($x_1 \in C_1$, $x_1 \notin C_2 $).  It is possible for 
$\lim\limits_{n \to \infty} 
\frac{\PoF^\Rational_{G'}}{\PoF^\Rational_G} = \infty$.
\end{theorem}

\begin{proof}

Consider a graph $G$ with two components: one component $K$ contains 2 vertices joint by an edge, the other component $S$ is a star with $s+1$ vertices ($s \geq 1/p$).  There are two groups: $C_1$ contains all degree-1 vertices from $S$ and one vertex from $K$; $C_2$ contains the other vertex $x_1$ from $K$ and the central vertex $x_2$ from $S$.  There is one seed ($k=1$), and the fair allocation of seeds to each group is $k_1=k_2=1$.

\begin{figure}[h]
\centering
\begin{minipage}{0.49\columnwidth}
\centering
\includegraphics[width=0.6\columnwidth]{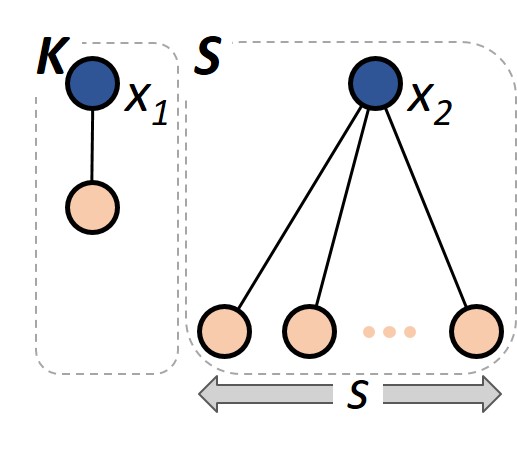}
\caption*{$G$ with Disjoint Groups.}
\end{minipage}
\begin{minipage}{0.49\columnwidth}
\centering
\includegraphics[width=0.6\columnwidth]{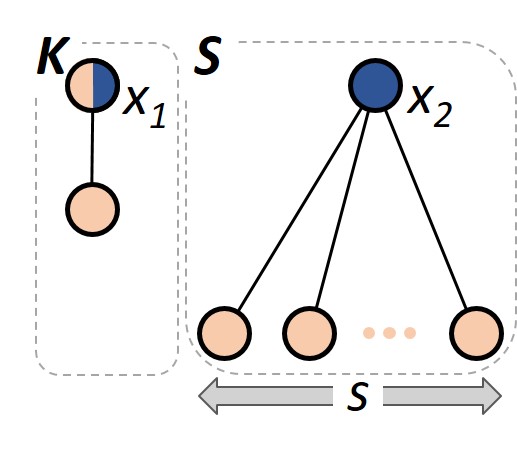}
\caption*{$G'$ with Overlapping Groups.}
\end{minipage}
\end{figure}

\begin{figure*}[ht!]
    \centering    \includegraphics[width=1.6in]{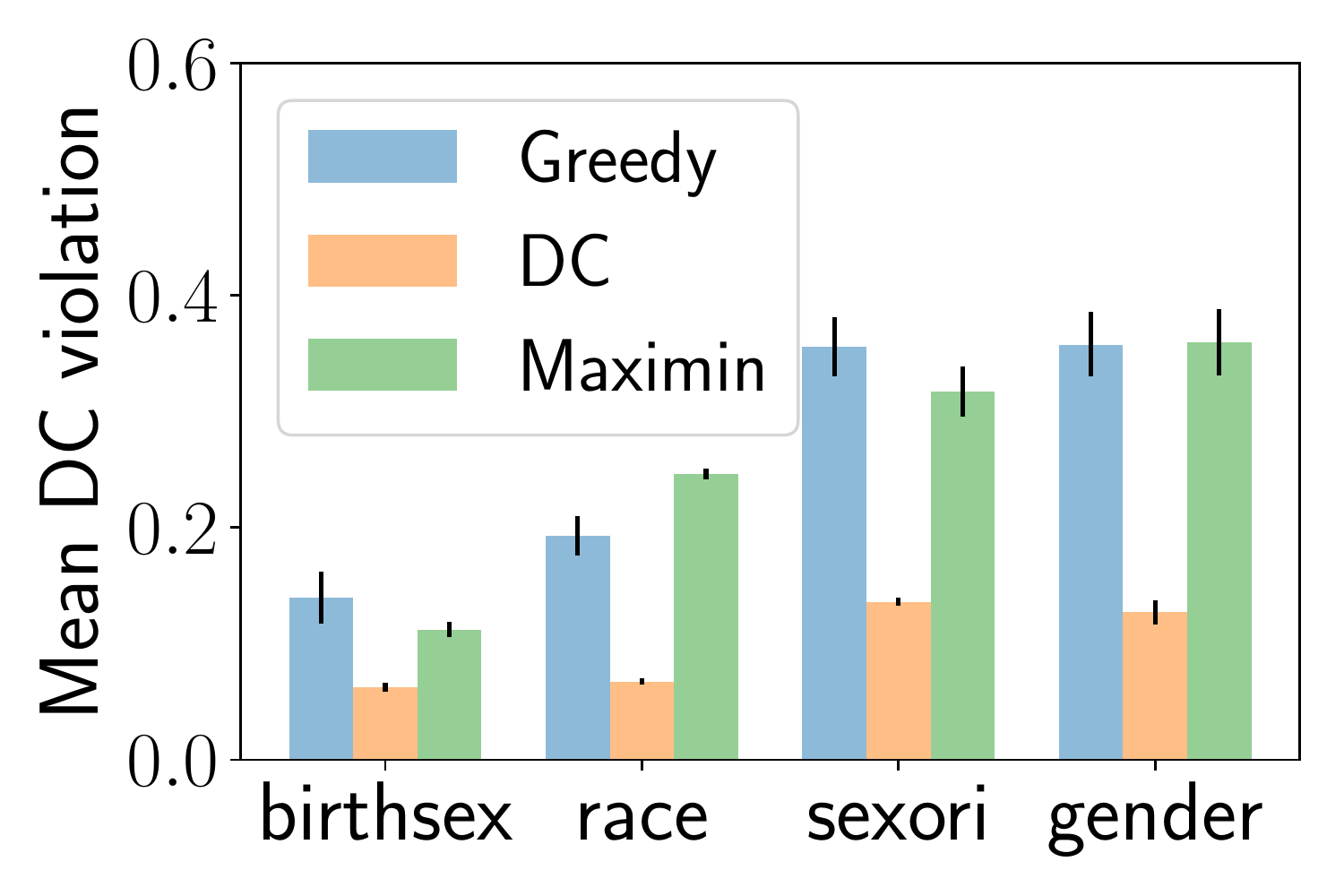}
    \includegraphics[width=1.6in]{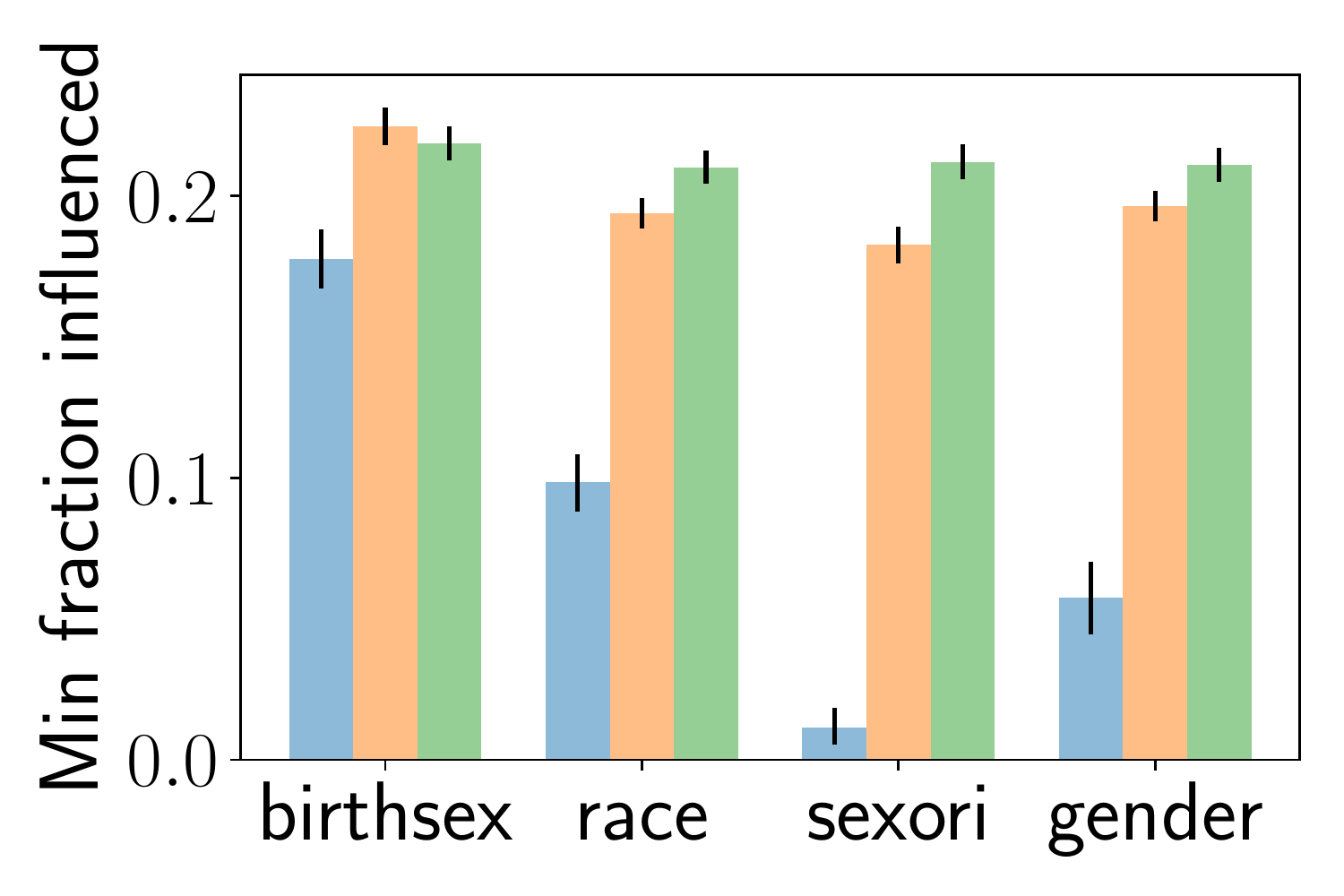}
    \includegraphics[width=1.6in]{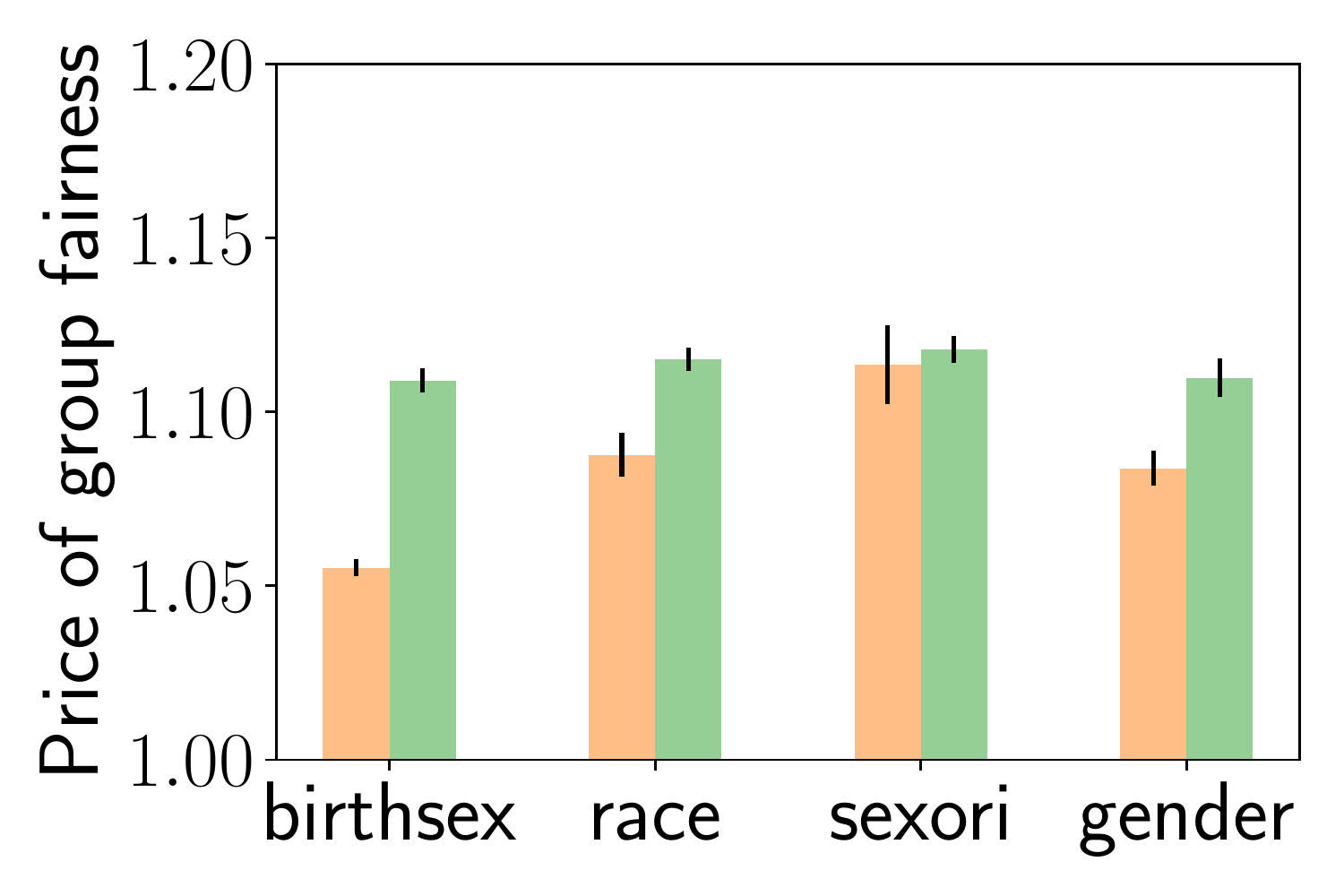}
    \includegraphics[width=1.6in]{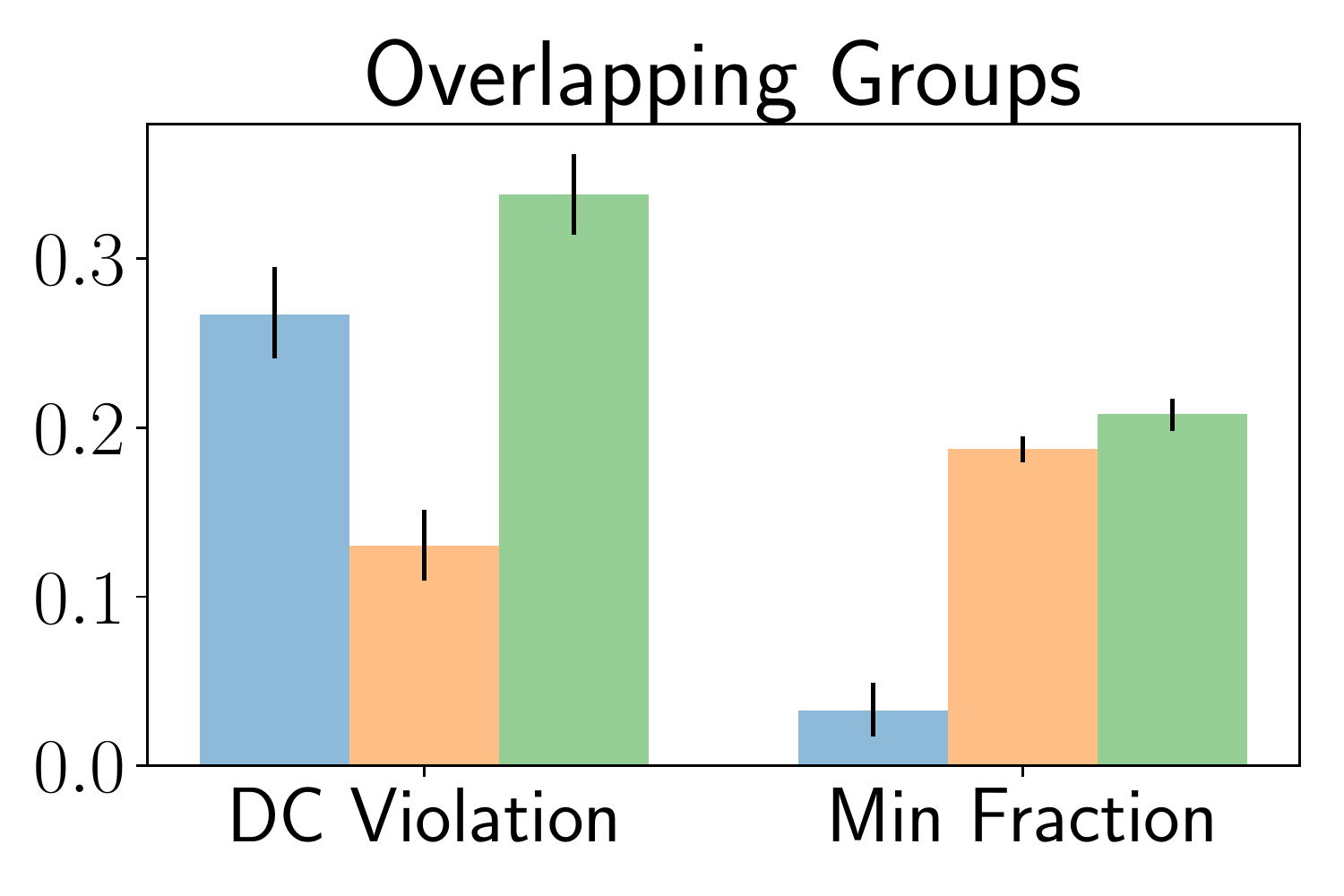}

    \includegraphics[width=1.6in]{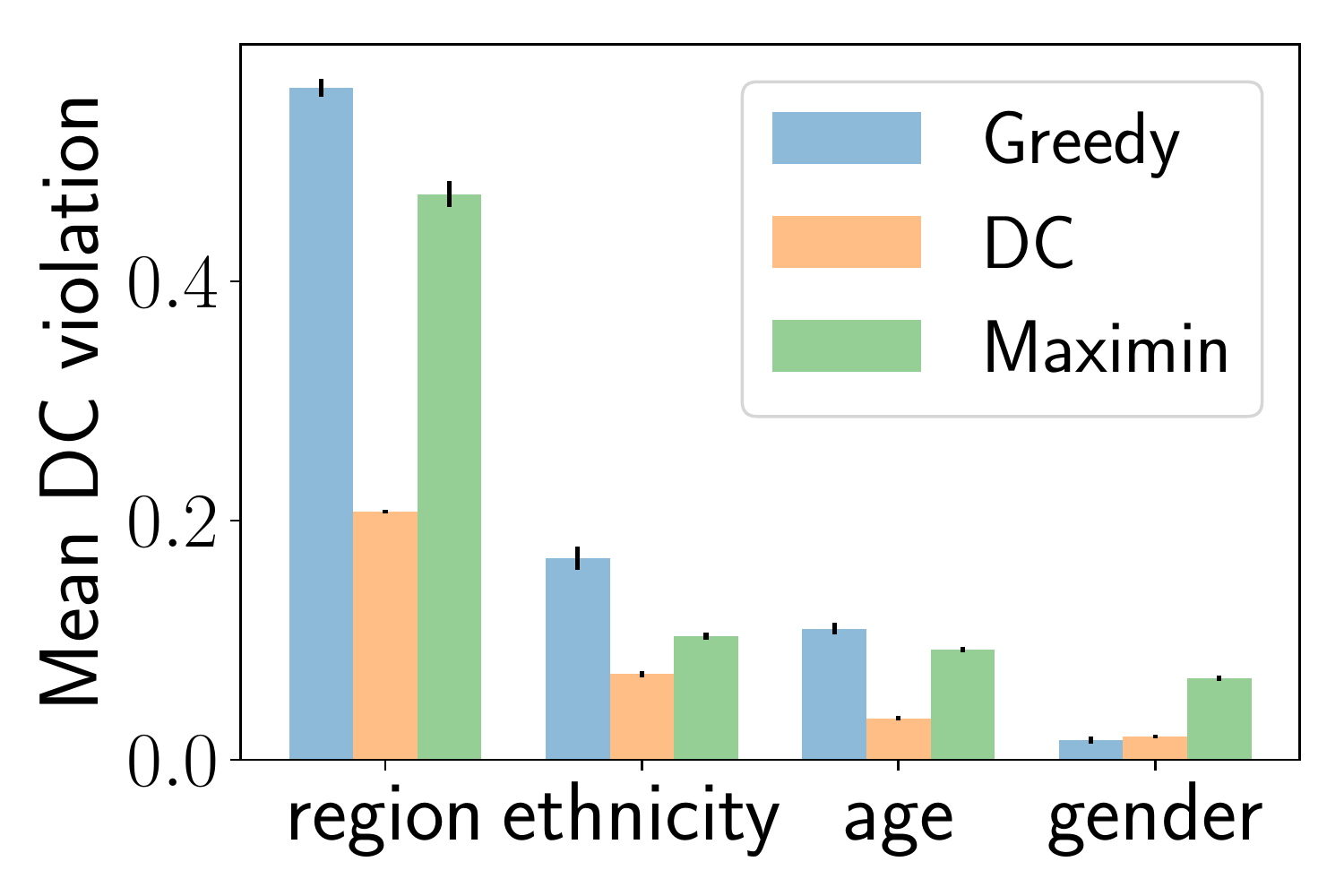}
    \includegraphics[width=1.6in]{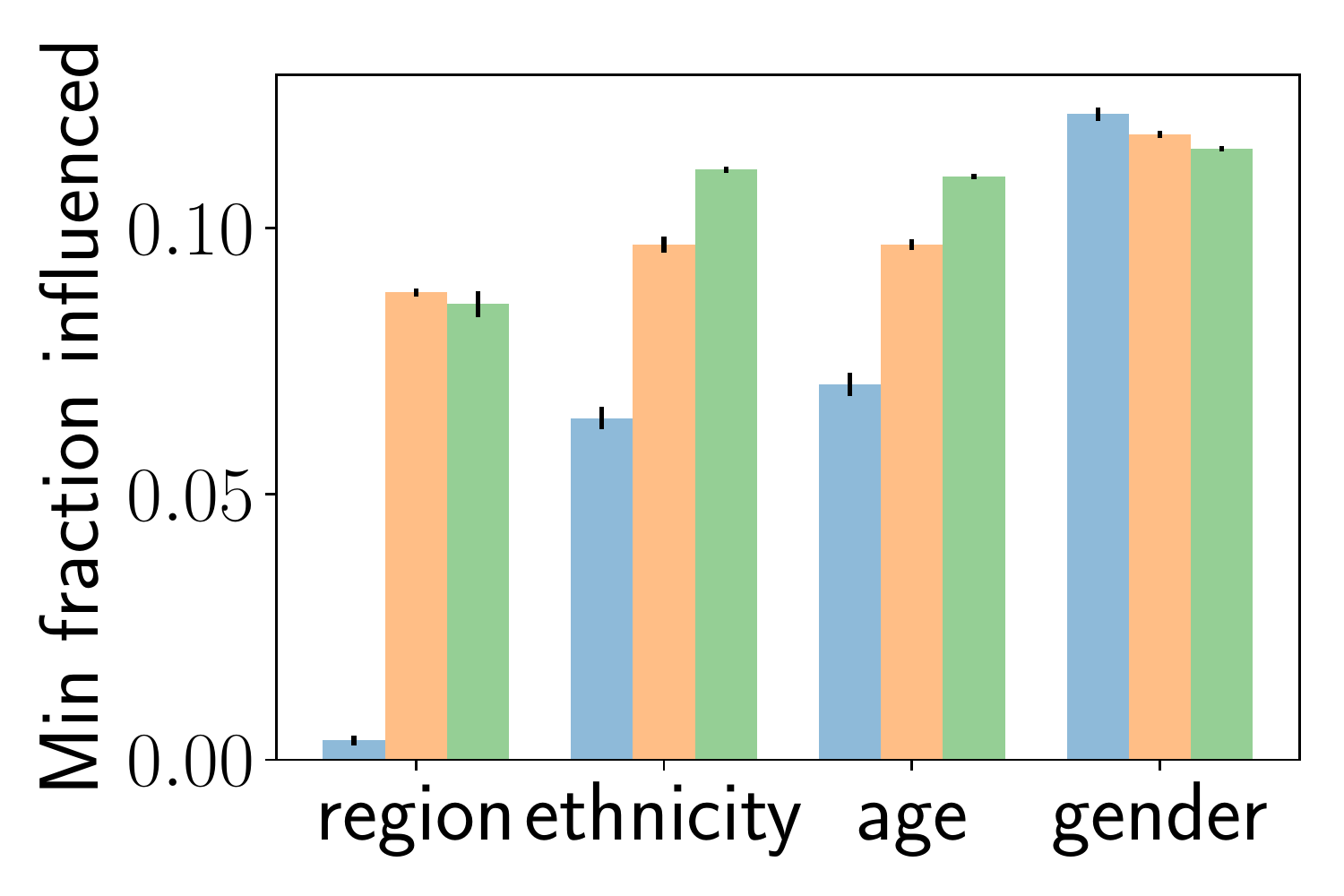}
    \includegraphics[width=1.6in]{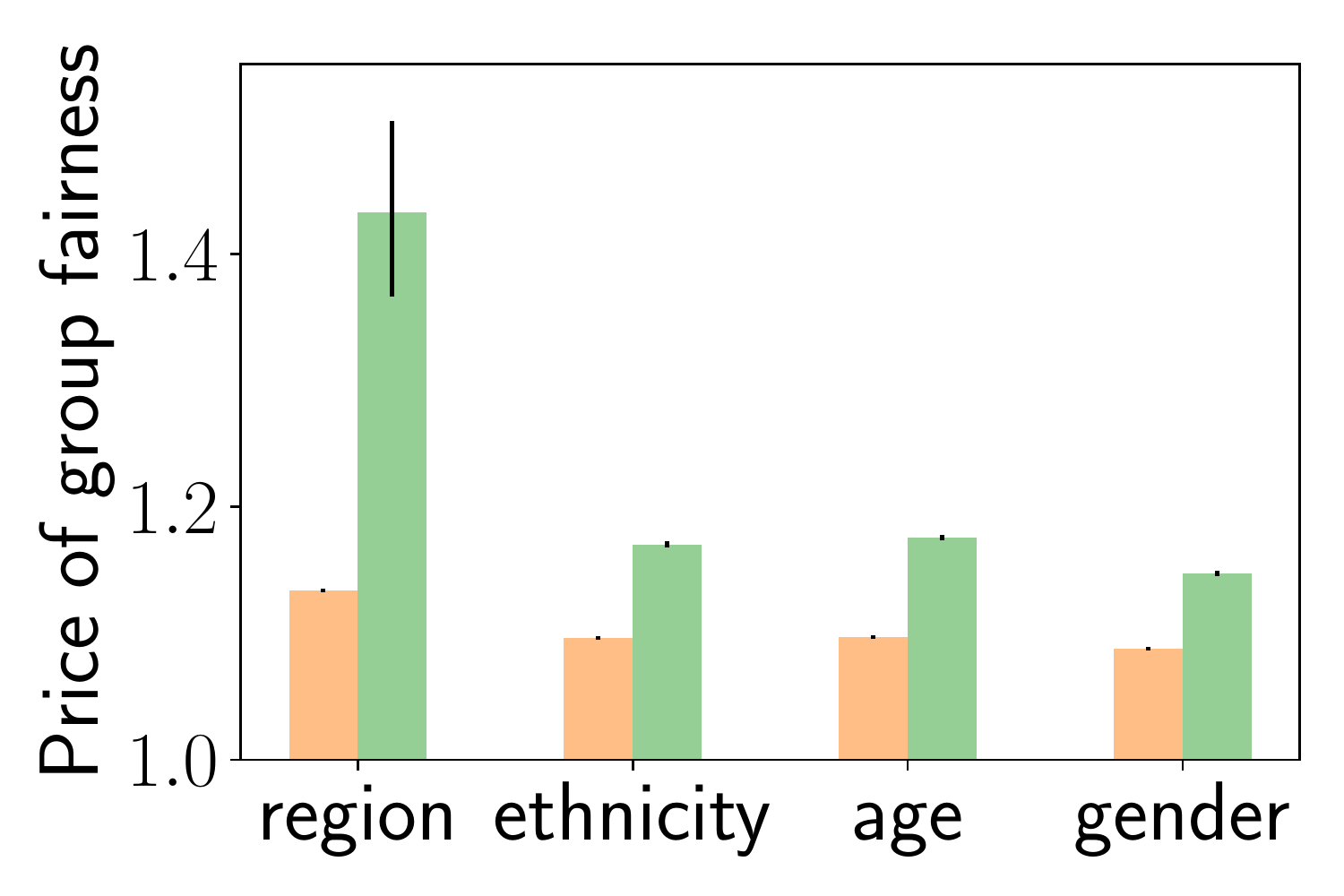}
    \includegraphics[width=1.6in]{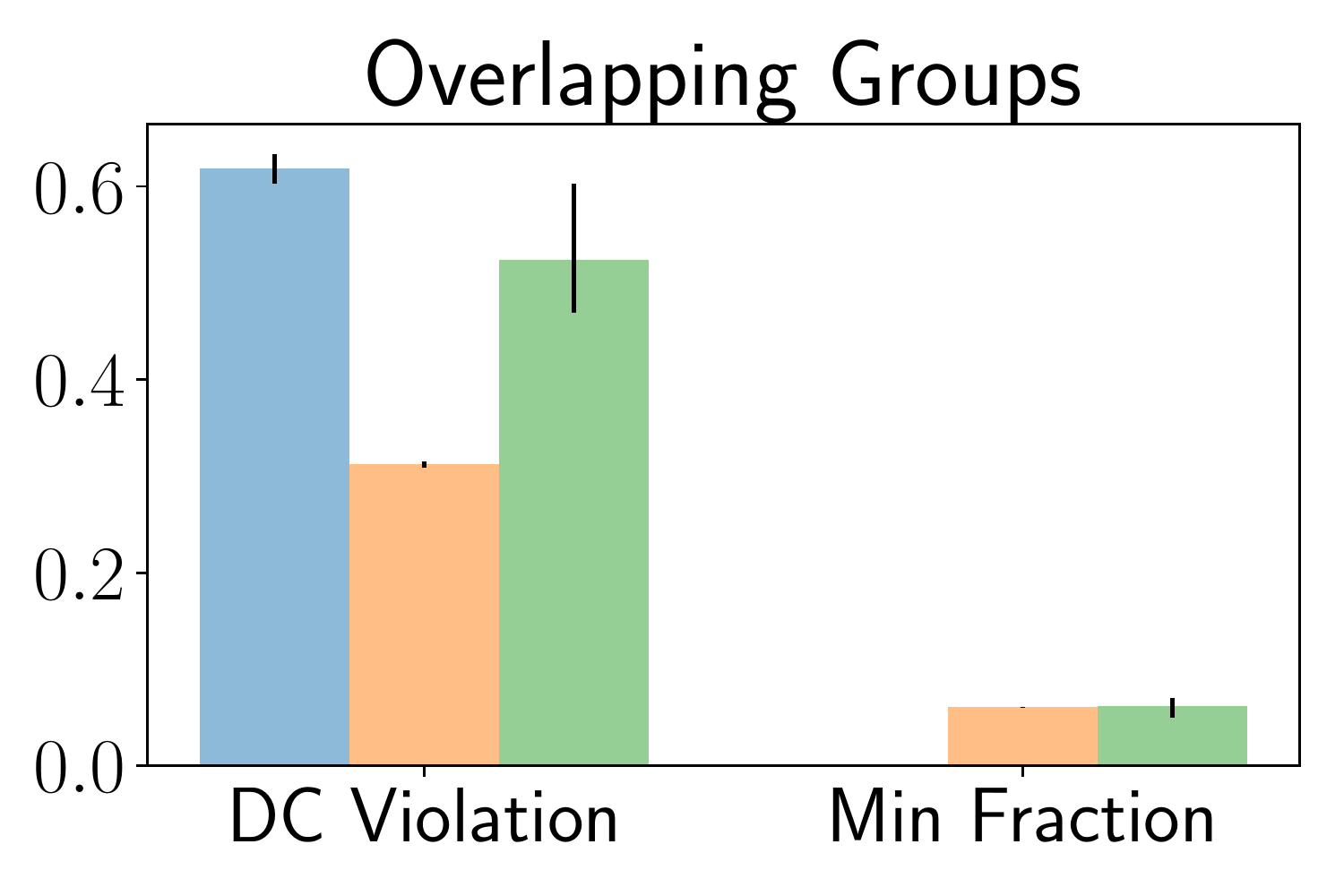}
    \caption{Average performance on homeless youth social networks (top) and simulated Antelope Valley networks (bottom).}
    \label{fig:hyh-spa}
\end{figure*}
Since the induced subgraphs for both groups comprise only of isolated nodes, the group rational influence for each group is $\Infl_{G[C_1]}=\Infl_{G[C_2]}=1$.  Therefore, the seed set $\{x_2\}$ is both fair and optimal, giving an expected influence of $\Infl_G(\{x_2\}) = 1+ps$.

Now, let us modify $G$ by letting $x_1$ belong to \emph{both} communities to obtain $G'$, and communities $C'_1$ and $C'_2$. The group rational influence for $C'_2$ remains the same (its members have not changed) but $\Infl_{G'[C'_1]}$ has increased to $1+p$ (by seeding $x_1$).  In fact, this forces the fair allocation to seed $x_1$ instead of $x_2$, for a fair influence of $\Infl_{G'}(\{x_1\}) = 1+p$.

As $n \to \infty$, $\lim\limits_{n \to \infty} 
\frac{\PoF^\Rational_{G'}}{\PoF^\Rational_G} = \lim\limits_{s \to \infty} \frac{1+ps}{1+p} = \infty$.
\end{proof}

A more technical construction can demonstrate a similar result for Maximin Fairness, but only as $p \to \frac{1}{3}^{-}$; that is, $p<\frac{1}{3}$ as $p$ approaches $\frac{1}{3}$.  The proof is provided in the appendix.

\begin{theorem} \label{thm:maximin-overlap}
Given graphs $G$ with groups $C_1$ and $C_2$, and $G'$ with groups $C'_1$ and $C'_2$, where $G'=G$ $C'_1 = C_1$ and $C'_2$ is obtained from $C_2$ by the addition of one vertex $x_1$ ($x_1 \in C_1$, $x_1 \notin C_2 $).  It is possible for 
$\lim\limits_{\substack{ n \to \infty \\ p \to \frac{1}{3}^{-}}} 
\frac{\PoF^\Maximin_{G'}}{\PoF^\Maximin_G} \to \infty$.
\end{theorem}

\section{Experimental results}

	

We now investigate the empirical impact of considering fairness in influence maximization. We start with experiments on a set of four real-world social networks which have been previously used for a socially critical application: HIV prevention for homeless youth. Each network has 60-70 nodes, and represents the real-world social connections between a set of homeless youth surveyed in a major US city. Each node in the network is associated with demographic information: their birth sex, gender identity, race, and sexual orientation. Each demographic attribute gives a partition of the network into anywhere from 2 to 6 different groups. For each partition, we compare three algorithms: the standard greedy algorithm for influence maximization, which maximizes the total expected influence (Greedy), Algorithm \ref{alg:multi-submodular} used to enforce diversity constraints (DC), and Algorithm \ref{alg:multi-submodular} used to find a maximin fair solution (Maximin). We set the propagation probability to be $p = 0.1$ and fixed $k = 15$ seeds (varying these parameters had little impact). We average over 30 runs of the algorithms on each network (since all of the algorithms use random simulations of influence propagation), with error bars giving bootstrapped 95\% confidence intervals. 

Figure~\ref{fig:hyh-spa}~(top) shows that the choice of solution concept has a substantial impact on the results. For the diversity constraints case, we summarize the performance of each algorithm by the mean percentage violation of the constraints over all groups. For the maximin case, we directly report the minimum fraction influenced over all groups. We see that greedy generates substantial unfairness according to either metric: it generates the highest violations of diversity constraints, and has the smallest minimum fraction influenced. Greedy actually obtains near-zero maximin value with respect to sexual orientation. This results from it assigning one seed to a minority group in a single run and zero in others. 

DC performs well across the board: it reduces constraint violations by approximately 55-65\% while also performing competitively with respect to the maximin metric (even without explicitly optimizing for it). As expected, the Maximin algorithm generally obtains the best maximin value. DC actually attains slightly better maximin value for one attribute (birthsex); however, the difference is within the confidence intervals and reflects slight fluctuations in the approximation quality of the algorithms. However, Maximin performs surprisingly poorly with respect to diversity constraint violations. This indicates that optimizing exclusively for equal influence spread may force the algorithm to focus on poorly connected groups which exhibit severe diminishing returns. DC is able to attain almost as much influence in such groups but is then permitted to focus its remaining budget for higher impact. Interestingly, the price of fairness is relatively small for both solution concepts, in the range 1.05-1.15 (though it is higher for maximin than for DC). This indicates that while standard influence maximization techniques can introduce substantial fairness violations, mitigating such violations may be substantially less costly in real world networks than the theoretical worst case would suggest. 

Finally, the rightmost plot in the top row of Figure~\ref{fig:hyh-spa} explores an example with overlapping groups. Specifically, we consider the race and birthsex attributes so that each node belongs to two groups. Constraint violations are somewhat higher than for either attribute individually, but the price of fairness remains small (1.07 for DC and 1.13 for Maximin). 

In Figure~\ref{fig:hyh-spa}~(bottom), we examine 20 synthetic networks used by Wilder et al. \shortcite{wilder2018optimizing} to model an obesity prevention intervention in the Antelope Valley region of California. Each node in the network has a geographic region, ethnicity, age, and gender, and nodes are more likely to connect to those with similar attributes. Each network has $500$ nodes and we set $k=25$. Overall the results are similar to the homeless youth networks. One exception is the high price of fairness that maximin suffers with respect to the ``region" attribute (over 1.4), but the other $PoF$ values are relatively low (below 1.2). We also observe that greedy obtains the (slightly) best maximin performance for gender, likely because the network is sufficiently well-mixed across genders that fairness is not a significant concern (as confirmed by the extremely low DC violations). Absent true fairness concerns, greedy may perform slightly better since it solves a simpler optimization problem.  However, in the last figure, we examine overlapping groups given by region and ethnicity  and observe that greedy actually obtains zero maximin value, indicating that there is one group that it never reached across any run. 

\section{Conclusions}

In this paper, we examine the problem of selecting key figures in a population to ensure the fair spread of vital information across all groups.  This problem modifies the classic influence maximization problem with additional fairness provisions based on legal and game theoretic concepts.  We examine two methods for determining these provisions, and show that the ``Price of Fairness'' for these provisions can be unbounded.  We propose an improved algorithm for multiobjective maximization to examine this problem on real world data sets.  We show that  standard influence maximization techniques often neglect smaller groups, and a diversity constraint based algorithm can ensure these groups receive a fair allocation of resources at relatively little cost.  As automated techniques become increasingly prevalent in society and governance, our technique will help ensure that small and marginalized groups are fairly treated.








\newpage
\bibliographystyle{named}
\bibliography{abbshort,firstnameshortbib}

\section{Appendix}
\appendix

\renewcommand*{\thetheorem}{\Alph{theorem}}





\maketitle

\section{Price of fairness}

\begingroup
\def\thetheorem{\ref{thm:maximin_pof}}
\begin{theorem}
$\PoF^{M}$ is unbounded.
\end{theorem}
\addtocounter{theorem}{-1}
\endgroup

\begin{proof}
Consider a graph $G$ with two components: $K$ which consists of 2 connected vertices, and $S$ which is a star with $s+1$ nodes.  Let the first group $C_1$ have only one node in $K$.  All remaining nodes belong to the second group $C_2$, including one node $x_1$ in $K$ and the central node of the star $x_2$.  We have $k=1$ seed.

\begin{SCfigure}[][h]
\centering
\includegraphics[width=0.4\columnwidth]{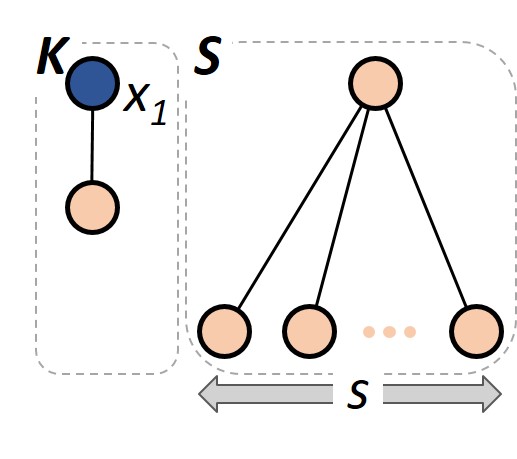}
\caption*{Example Undirected Network with Unbounded PoF under Maximin Fairness}
\end{SCfigure}

It is clear that the optimal seeding configuration is to seed $x_2$, which gives $\Infl^\OPT=1+ps$.  However, this is not a maximin fair seeding, as $C_1$ receives 0 influence.  Instead, seeding $x_1$ is maximin fair, giving $C_1$ $p$ influence and $C_2$ $1$ influence, giving a maximin utility $U^\Maximin(\{x_1\}) = \min(p,\frac{1}{s+2})$.  In this case, $\Infl^\Maximin = 1+p$.

As $s \to \infty$, $PoF_\Maximin = \frac{1+ps}{1+p}$ becomes unboundedly large.
\end{proof}

\begingroup
\def\thetheorem{\ref{thm:not-submod}}
\begin{theorem}
$U^\Maximin$ and $U^\Rational$ are not submodular.
\end{theorem}
\addtocounter{theorem}{-1}
\endgroup

We divide the proof of this theorem into two parts:

\begin{conjecture} \label{thm:not-submod-partA}
Maximin utility $U^\Maximin$ is not submodular.
\end{conjecture}

\begin{proof}

Let us consider a graph with 4 nodes $\{x,a,b,c\}$ where $\{x,a\}$ form community $C_1$ and $\{b,c\}$ form community $C_2$.  Let $A = \{a,b\}$ and $B= \{a,b,c\}$ be two possible seeding configurations.

\begin{SCfigure}[][h]
\centering
\includegraphics[width=0.4\columnwidth]{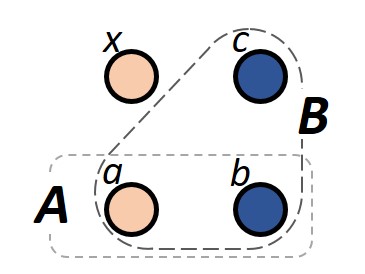}
\caption*{Example showing non-submodularity of Maximin Fairness}
\end{SCfigure}

Notice that $C_1$ receives $1$ influence in both configurations, which is weakly less than the influence received by $C_2$, and so, $U^\Maximin(A) = U^\Maximin(B) = 1/2$.

Now, consider adding $x$ to the $A$ and $B$. $U^\Maximin(A\cup\{x\})=1/2$ since $C_2$ remains incompletely seeded.  But $U^\Maximin(B\cup\{x\})=1$ since both groups are fully seeded.
\end{proof}

\begin{conjecture}
Group rational utility $U^\Rational$ is not submodular.
\end{conjecture}

\begin{proof}

Recall the definition of group rational utility: 

\begin{align*}
U^\Rational(A) =~
\begin{cases}
	\Infl_G(A), &if~constraints~satisfied \\
	0, &otherwise.
\end{cases}	
\end{align*}

Let us consider the same graph as in Conjecture~\ref{thm:not-submod-partA} with 4 nodes $\{x,a,b,c\}$ where $\{x,a\}$ form community $C_1$, and $\{b,c\}$ forms community $C_2$.  $k=4$ seeds are available, and so therefore the group rational constraints are only satisfied by seeding all vertices.

Let $A=\{a, b\}$ and $B=\{a,b,c\}$.  It is easy to verify that $U^\Rational(A)=U^\Rational(B)=U^\Rational(A \cup \{x\}) = 0$ since none of these satisfy all group rational constraints.  However, $U^\Rational(B \cup \{x\}) > 0$, and so therefore $f(A \cup \{x\}) - f(A) < f(B \cup \{x\}) - f(B)$ for $A \subseteq B$, which contradicts the definition of submodularity. 

\end{proof}


\begingroup
\def\thetheorem{\ref{thm:maximin-overlap}}
\begin{theorem}
Give graphs $G$ with groups $C_1$ and $C_2$, and $G'$ with groups $C'_1$ and $C'_2$, where $G'=G$ $C'_1 = C_1$ and $C'_2$ is obtained from $C_2$ by the addition of one vertex $x_1$ ($x_1 \in C_1$, $x_1 \notin C_2 $.  It is possible for 
$\lim\limits_{\substack{ n \to \infty \\ p \to \frac{1}{3}^{-}}} 
\frac{\PoF^\Maximin_{G'}}{\PoF^\Maximin_G} \to \infty$.
\end{theorem}
\addtocounter{theorem}{-1}
\endgroup

\begin{proof}
Consider a graph $G$ with two star components: $S_1$ with $s+1$ vertices with a central node $x_1$, and $S_2$ with $t+2$ vertices with central node $x_2$ ($s>t$).  There are two groups: $C_1$ contains 2 vertices, $x_1$ and a non-central node from $S_2$; $C_2$ contains $s+t+1$ remaining vertices, including $x_2$.  There is one seed ($k=1$), and a total of $n=s+t+3$ nodes.

\begin{figure}[h]
\centering
\begin{minipage}{0.49\columnwidth}
\centering
\includegraphics[width=0.95\columnwidth]{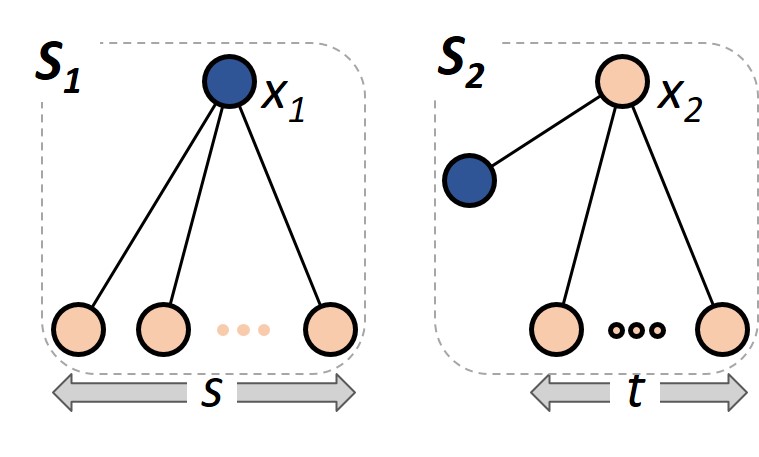}
\caption*{$G$ with Disjoint Groups.}
\end{minipage}
\begin{minipage}{0.49\columnwidth}
\centering
\includegraphics[width=0.95\columnwidth]{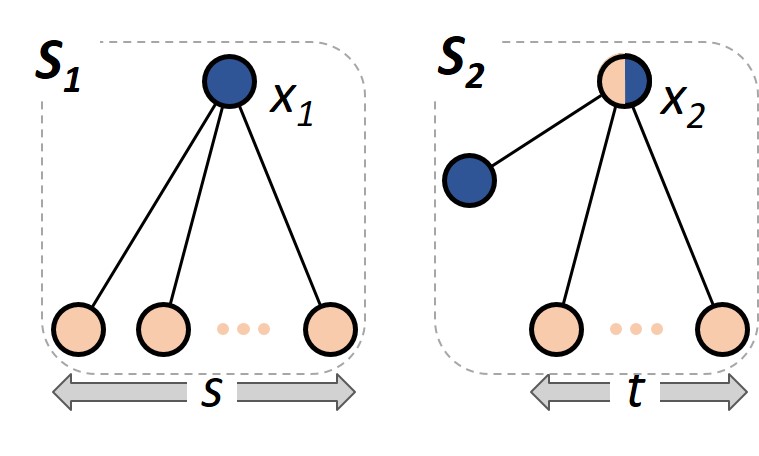}
\caption*{$G'$ with Overlapping Groups.}
\end{minipage}
\end{figure}

It is easy to see that the Maximin configuration is to seed $x_1$, which gives $C_1$ $1$ influence, and $C_2$ $ps$ influence.  This gives a Maximin influence $\Infl^\Maximin_G = 1+ps$.\footnote{We do not need to calculate $U^\Maximin$ explicitly at any point in this proof as it is not required for the proof to work.}

Now, consider a modified graph $G'=G$, but with our groups modified by allowing $x_2$ to belong to both communities.  That is, $C'_1 = C_1$ and $C'_2 = C_2 \cup \{x_2\}$.  The Maximin configuration has two possibilities: either $\{x_1\}$ remains the Maximin configuration, or $\{x_2\}$ becomes the new Maximin configuration.  In order for the latter case to be true, seeding $\{x_2\}$ must provide higher proportional influence to the least well-off group than seeding $\{x_1\}$.

Seeding $\{x_2\}$ generates $\frac{1+p}{3}$ influence for $C_1$, and $\frac{1+pt}{s+t+1}$ influence for $C_2$.  Seeding $\{x_1\}$ generates $\frac{1}{3}$ influence for $C_1$, and $\frac{ps}{s+t+1}$.  It can be shown that for $p<\frac{1}{3}$ and $t=\frac{s}{1-3p}$, these conditions are satisfied and $\{x_2\}$ is the Maximin configuration, generating a total of $\Infl^\Maximin_{G'} = 1+p(t+1)$.

Then,
\begin{align*}
\lim_{n \to \infty} \frac{\PoF^{\Maximin}(G')}{\PoF^{\Maximin}(G)}
&= \lim_{n \to \infty} \frac{\Infl^{\Maximin}(G)}{\Infl^{\Maximin}(G')} \\
&= \lim_{n \to \infty} \frac{1+sp}{1+p(t+1)} \\
&= \lim_{n \to \infty} \frac{1+sp}{1+p(\frac{s}{1-3p}+1)} \\
&= 1-3p
\end{align*}

And therefore, as $p \to \frac{1}{3}^-$, i.e. $p$ approaches $\frac{1}{3}$ \emph{from the left}, the addition of a node to a second group may cause the Price of Maximin Fairness to worsen by an arbitrarily large amount.
\end{proof}

	\section{Analysis of multiobjective submodular maximization problem}
	
	Consider a collection of monotone submodular functions $f_1...f_m$ with corresponding multilinear extensions $F_1...F_m$. We will assume that the maximum singleton value of any item in the ground set $V$ is bounded as $f_i{\{v\}} \leq b$ for all $i \in [m], v \in V$. Suppose that we are given a target value $W_i$ for each $f_i$ and would like to find a set $S$ with $|S| \leq k$ which guarantees $f_i(S) \geq W_i$ for all $i$. We are promised that such an $S$ exists. We will give an approximation algorithm for this problem which improves in terms of both runtime and approximation ratio on the best current algorithms, given by Udwani \cite{udwani2018multi}, who in turn build on the work of Chekuri et al.\ \cite{chekuri2010dependent}.

	Our algorithm follows the overall template of \cite{udwani2018multi}, which carries out three steps (given a precision level $\epsilon$). 
	
	\begin{enumerate}
		\item Make a pass over the ground set, maintaining a set $S_1$. Add to $S_1$ every item which has value at least $\epsilon^3 W_i$ for some $f_i$. 
		\item Define $\polym$ to be the uniform matroid polytope for budget $k - |S_1|$. Use a subroutine to find a point $x \in \polym$ satisfying $F_i(x|x_{S_1}) \geq \alpha \left(W_i - f_i(S_1)\right) - \epsilon$ for all $i$ and some approximation ratio $\alpha$. This is the key step where we improve the runtime and approximation ratio. 
		\item Round $x$ to a set $S_2$ using the swap rounding algorithm of \cite{chekuri2010dependent} Output $S_1 \cup S_2$. 
	\end{enumerate}

	Our primary technical contribution is an algorithm for the second step which guarantees $\alpha = \left(1 - \frac{1}{e}\right)$. It uses access to three kinds of stochastic oracles for the functions and their multilinear extensions:


			
			
			
			

	\begin{enumerate}
		\item A stochastic value oracle for singletons $\mathcal{A}^i_{\text{val}}$ corresponding to each $f_i$. Given an item $v$, this oracle returns a value $\mathcal{A}^i_{\text{val}}(v)$ with $\E\left[\mathcal{A}^i_{\text{val}}(S)\right] = f_i(\{v\})$ and $\text{Var}\left[\mathcal{A}^i_{\text{val}}(S)\right] \leq c_{\text{val}}$.
		\item A stochastic gradient oracle $\mathcal{A}^i_{\text{grad}}$ for each multilinear extension $F_i$. Given a point $x \in \polym$, $\mathcal{A}^i_{\text{grad}}(x)$ satisfies $\E\left[\mathcal{A}^i_{\text{grad}}\right] = \nabla_x F_i(x)$ and $\norm{\mathcal{A}^i_{\text{grad}}(x)}_\infty \leq c_{\text{grad}}$
		\item A stochastic gradient oracle $\mathcal{A}^j_{\text{item}}$ corresponding to each item $j \in [n]$. Given a point $x \in \polym$,  $\mathcal{A}^j_{\text{item}}(x)$ satisfies $\E\left[\mathcal{A}^j_{\text{item}}(x)\right] = \left[\nabla_{x_j} F_1(x) ... \nabla_{x_j} F_m(x)\right]$ and $\norm{\mathcal{A}^j_{\text{item}}(x)}_\infty \leq c_{\text{item}}$. Note that this can be simulated from the above oracle, but may sometimes admit more efficient implementations. 
	\end{enumerate}

	
	We now analyze this algorithm. We start by recalling a technical lemma on the smoothness of the multilinear extension:
	
	\begin{lemma}[Hassani et al. \cite{hassani2017gradient}, Lemma C.1] \label{lemma:smooth}
		For any monotone submodular set function $f$ and its multilinear extension $F$, $||\nabla F(x) - \nabla F(y)||_\infty \leq b||x - y||_1$ where $b = \max_{v \in V} f(\{v\})$. That is, $F$ is $b$-smooth with respect to the $\ell_1$ norm. 
	\end{lemma}

	\begin{lemma}
		$F$ is $b$-Lipschitz in the $\ell_1$ norm. \label{lemma:lipschitz}
	\end{lemma}
	\begin{proof}
		Recall that $\nabla_{x_j} F(x) = \E_{S\sim x}[f(S \cup \{j\}) - f(S \setminus\{j\})]$ CITE, where $S\sim x$ denotes including each $j$ in $S$ independently with probability $x_j$. By submodularity, $\E_{S\sim x}[f(S \cup \{j\}) - f(S \setminus\{j\})] \leq f(\{j\}) \leq b$. Hence, $||\nabla_{x_j} F(x)||_\infty \leq b$ which proves the lemma.  
	\end{proof}
	
	Next, we show a guarantee for the output of mirror descent in step 2(a).

	\begin{lemma} \label{lemma:mirror}
		For some $x \in \polym$, suppose that there exists a $v^* \in \polym$ such that $v^* \cdot \nabla F_i(x) \geq W_i - F_i(x)$ for all $i = 1...m$. Then, S-SP-MD returns a $v$ satisfying $v \cdot \nabla F_i(x) \geq (1 - \epsilon)(W_i - F_i(x)) - \epsilon$ for all $i$ with probability $1-\delta$. There are $O\left(\frac{\left(c_{\text{grad}}\sqrt{k \log n} + kc_{\text{item}}\sqrt{\log n}\right)^2}{\epsilon^4} \log \frac{1}{\delta}\right)$ iterations, each requiring one call to oracles $\mathcal{A}^i_{\text{grad}}$ and $\mathcal{A}^j_{\text{item}}$ for some $i$ and $j$, and $O(n+m)$ additional work. 
	\end{lemma}

	\begin{proof}
	
	    Our objective is to find a $v$ satisfying $v \cdot F_i(x) \geq (1 - \epsilon)(W_i - F_i(x)) - \epsilon$, under the guarantee that such a $v$ exists. Note that we call S-SP-MD only on the set of indices $\mathcal{I}$ where $W_i - F_i(x) \geq \epsilon$. For all other indices, where the current solution is already within $\epsilon$ of the target, monotonicity of the $F_i$ guarantees that  $v \cdot F_i(x) \geq 0 \geq W_i - F_i(x) - \epsilon$. 
	    
	    The feasibility problem on the groups in $\mathcal{I}$ is equivalent to solving maxmin problem 
	    
	    \begin{align*}
	        \max_{v \in \mathcal{P}} \min_{i \in \mathcal{I}} \frac{v \cdot \nabla F_i(x)}{W_i - F_i(x)}
	    \end{align*}
	    
	    To see this, let $OPT$ denote the optimal value for the maxmin problem; we are guaranteed $OPT \geq 1$. If we have $v$ with maxmin value at least $OPT - \epsilon$, then $v$ satisfies
	    
	    \begin{align*}
	        v \cdot \nabla F_i(x) \geq (1 - \epsilon)(W_i - F_i(x)) \quad \forall i \in \mathcal{I}
	    \end{align*}
		
		We now prove that S-SP-MD produces a $v$ with maxmin value at least $OPT-\epsilon$. Let $A$ be a matrix where column $i$ is $\frac{\nabla F_i(x)}{W_i - F_i(x)}$ for each $i \in \mathcal{I}$, and define $g(v, y) = v^\top A y$. Let $\Delta(\mathcal{I})$ be the $|\mathcal{I}|$-dimensional probability simplex. We would like to solve the problem 
		
		\begin{align*}
		\max_{v \in \polym} \min_{y \in \Delta(\mathcal{I})} g(v,y)
		\end{align*} 
		
		which is easily seen to be equivalent to the original maxmin problem. 
		
		
		We will solve the above saddle point problem by running stochastic saddle point mirror descent with the negative entropy mirror map on the function $g$. We obtain stochastic estimates of $\nabla_v g(v,y)$ and $\nabla_y g(v,y)$ via calls to input the oracles. First, note that
		
		\begin{align*}
		    \nabla_v g(v,y) &= Ay  = \sum_{i \in \mathcal{I}} y_i A_{\cdot, i} = \E_{i \sim y}\left[A_{\cdot, i}\right]
		\end{align*}
		
		where $i \sim y$ denotes drawing index $i$ with probability $y_i$ (recall that $y \in \Delta(\mathcal{I})$ is a probability distribution). Hence, we can obtain an estimate $\hat{\nabla_v}$ of $\nabla_v g(v, y)$ by sampling $i \sim y$ and returning $\frac{1}{W_i - F_i(x)}\oracle{i}{grad}$. We are guaranteed $\norm{\hat{\nabla_v}}_\infty \leq \frac{c_{\text{grad}}}{W_i - F_i(x)} \leq \frac{c_{\text{grad}}}{\epsilon}$. We take a similar strategy for $\nabla_y g(v,y)$: $v^\top A = k \left(\frac{1}{k}v\right)\hat{A} = k\E_{j \sim \frac{1}{k}v}[v_jA_j]$ (since $\frac{1}{k}v_j$ is a probability distribution). Hence, we can sample $j \sim \frac{1}{k}v$ and return $\hat{\nabla}_y = k \cdot \text{diag}\left(\frac{1}{\vec{W} - \vec{F}(x)}\right)\oracle{j}{item}(x)$. This satisfies $\norm{\hat{\nabla}_y}_\infty \leq \frac{k}{\epsilon} c_{\text{item}}$.


		
		Note that we can bound the diameter of $\polym$ with respect to the mirror map by $\sqrt{k \log n}$ (see \cite{hassani2017gradient}) and the diameter of $\Delta^m$ by $\sqrt{\log m}$ (see \cite{nemirovski2009robust}). We will run mirror descent for $T'$ iterations. Let $\bar{x} = \frac{1}{T'}\sum_{t = 1}^{T'} x^t$ and $\bar{y} = \frac{1}{T'}\sum_{t = 1}^{T'} y^t$. Now applying Proposition 3.2 of Nemirovski et al.\ \cite{nemirovski2009robust} implies that after $T'$ iterations we have
		
		\begin{align*}
		&\Pr\Bigg[\max_{v \in \polym} g(v, \bar{y}) - \min_{y \in \Delta(\mathcal{I})} g(\bar{v}, y) \geq \\
		& \frac{(8 + 2\Omega)\sqrt{5}\left(c_{\text{grad}}\sqrt{k \log n} + k c_{\text{item}}\sqrt{\log n}\right)}{\epsilon\sqrt{T}}\Bigg]  \leq 2\exp(-\Omega)
		\end{align*}
		
		and so taking $T' = O\left(\frac{\left(c_{\text{grad}}\sqrt{k \log n} +k c_{\text{item}}\sqrt{\log n}\right)^2 \log\frac{1}{\delta}}{\epsilon^4}\right)$ ensures that 
		
		\begin{align*} 
		\min_{y \in \Delta(\mathcal{I})} g(\bar{v}, y) \geq \max_{v \in \polym} \min_{y \in \Delta(\mathcal{I})} g(v,y) - \epsilon.
		\end{align*}
		
		holds with probability at least $1-\delta$.

%
%

%
%
%
%
	\end{proof}

	\begin{theorem}
	Suppose that there exists some $x\in \polym$ satisfying $F_i(x) \geq W_i$ for all $i = 1...m$. Then, after $T = \frac{bk^2}{\epsilon}$ iterations, the algorithm returns a point $x^T$ satisfying $F_i(x^T) \geq \left(1 - \epsilon\right)\left(1 - \frac{1}{e}\right)W_i - \epsilon$ for all $i$. Each iteration requires one call to mirror descent at success probability $\delta' = \frac{\delta \epsilon}{bk^2}$ and precision level $\epsilon' = \frac{\epsilon}{2}$, $O(m)$ $\epsilon$-accurate value oracle calls, and $O(n)$ additional work. \label{theorem:fw}
	\end{theorem}
	\begin{proof}
	We analyze the progress that the algorithm makes with respect to each $F_i$ over a single step $t$. Using the guarantee for the subroutine mirror descent (run with a precision level $\epsilon_1$ to be set below), and assuming that the values $\{W_i\}$ are feasible, we have with probability at least $1-\delta$
	\begin{align*}
	&F_i(x^t) - F_i(x^{t-1}) \\
	 &\geq \frac{1}{T}\left[\nabla F_i(x^{t-1}) \cdot v^t\right] - \frac{b}{2}\norm{x^{t} - x^{t-1}}_1^2 \text{(Lemma \ref{lemma:smooth})}\\
	&\geq \frac{1}{T}\left[\nabla F_i(x^{t-1}) \cdot v^t\right] - \frac{bk^2}{2T^2} \text{($\ell_1$ diameter of $\polym$)}\\
	&\geq \frac{1}{T}\left((1 - \epsilon_1)(W_i  - F_i(x^{t-1}))  - \epsilon_1\right) - \frac{bk^2}{2T^2} \text{(Lemma \ref{lemma:mirror})}
	\end{align*}
	
	which implies
	
	\begin{align*}
	W_i  - F_i(x^t)\leq \left(1 - \frac{1 - \epsilon_1}{T}\right)\left[W_i - F_i(x^{t-1})\right] + \frac{\epsilon_1}{T} + \frac{bk^2}{2T^2}
	\end{align*}
	
	and so after $T$ steps
	
	\begin{align*}
	W_i  - F_i(x^T) &\leq \left(1 - \frac{1-\epsilon_1}{T}\right)^T\left[W_i - F_i(x^{0})\right] + \epsilon_1 + \frac{bk^2}{2T}\\
	&\leq \frac{1}{e^{1-\epsilon_1}}W_i + \epsilon_1 + \frac{bk^2}{2T}
	\end{align*}

	holds with probability at least $1-T\delta$ via union bound. Taking $\epsilon_1 = \frac{\epsilon}{2}$, $T = \frac{bk^2}{\epsilon}$, and running mirror descent with success probability $\frac{\delta}{T}$ at each iteration ensures that
	
	\begin{align*}
	F_i(x^T) &\geq \left(1 - \frac{1}{e^{1 - \epsilon}}\right) W_i - \epsilon\\
	&\geq (1 - \epsilon)\left(1 - \frac{1}{e}\right) W_i - \epsilon
	\end{align*}
	
	holds for all $i$ with probability at least $1-\delta$, which completes the guarantee for the solution quality. To obtain the bound on additional work done by the algorithm, we note that the only operation performed besides calling mirror descent is adding $v^t$ to the current iterate, which takes time $O(n)$. 
	\end{proof}

	\begin{theorem}
		Given a feasible set of target values $W_1...W_n$, Algorithm \ref{alg:multi-submodular} outputs a set $S$ such that $f_i(S) \geq (1-\epsilon)\left(1 - \frac{m}{k(1 +\epsilon')\epsilon^3}\right)\left(1 - \frac{1}{e}\right)W_i - \epsilon$ with probability at least $1-\delta$. Asymptotically as $k\to\infty$, the approximation ratio can be set to approach $1 - 1/e$ so long as $m= o(k \log^3 k)$. The algorithm requires $O(nm)$ $\epsilon'$-accurate value oracle calls, $O(m\frac{bk^2}{\epsilon}\log \frac{1}{\delta})$ $\epsilon$-accurate value oracle calls, $O\left(\frac{bk^4 c^2}{\epsilon^5}  \log\left(n+ \frac{bk}{\delta\epsilon}\right)\right)$ calls to $\oracle{}{grad}$ and $\oracle{}{item}$, and $O\left(\frac{nk^2b^2}{\epsilon^2} + \frac{mk^2b}{\epsilon} + \frac{k^3 b^2}{\epsilon^2}\right)$ additional work. 
	\end{theorem}

	\begin{proof}
		\textsc{ThresholdInclude} produces a set $S_1$ for which each item $j \in S_1$ satisfies $f_i(\{j\}) \geq W_i(1 + \epsilon')\epsilon^3$ for some $i$, and any $j \not\in S_1$ satisfies $f_i(\{j\}) \leq W_i\epsilon^3$ for all $i$. Note that there can be at most $\frac{1}{(1 + \epsilon')\epsilon^3}$ items with $f_i(\{j\}) \geq W_i(1 + \epsilon')\epsilon^3$ for any given $i$ (combining submodularity with our WLOG assumption that $f_i$ is upper bounded by $W_i$). Hence, $|S_1| \leq  \frac{m}{(1 + \epsilon')\epsilon^3}$. Define $k_1 = k - |S_1|$.
		
		Now we lower bound the marginal gain of the fractional vector $x$ returned by \textsc{MultiobjectiveFW}. So long as the target values $\{\frac{k_1}{k}\left(W_i - f_i(S_1)\right)\}$ are feasible, we are guaranteed that $F_i(x|S_1) \geq \frac{k_1}{k}\left(1 - \frac{1}{e}\right)\left(W_i - f_i(S_1)\right) - \epsilon$. for all $i$. To see feasibility, let $S^*$ be the promised set satisfying the overall feasibility problem (i.e., $f_i(S^*) \geq W_i$ for all $i$). Let $x_{S}$ denote the indicator vector of the set $S$. We have that $|S^*\setminus S_1| \leq k$, and $F_i(x_{S^*\setminus S_1} | x_{S_1}) = f_i(S^*|S_1) \geq W_i - f_i(S_1)$. Using Corollary 3 of \cite{udwani2018multi}, the point $x' = \frac{k_1}{k}x_{S^*\setminus S_1}$ satisfies $F_i(x'|x_{S_1}) \geq \frac{k_1}{k}(W_i - f_i(S_1))$. $x'$ is also feasible for the continuous problem since $||x'||_1 \leq k_1$. Now applying Theorem \ref{theorem:fw} guarantees that  $F_i(x|S_1) \geq \frac{k_1}{k}\left(1 - \frac{1}{e}\right)\left(W_i - f_i(S_1)\right) - \epsilon$ with probability at least $1-\delta$. 
		
		Lastly, we need to handle the rounding process. We first take the point $x$ and approximately decompose it into a convex combination of integral points of $\mathcal{P}$. This is done using the algorithm of Mirrokni et al.\ \cite{mirrokni2017tight}, which produces a point $x_{\text{int}}$ satisfying $||x_{\text{int}} - x||_1 \leq \epsilon$ along with a decomposition of $x_{\text{int}}$ into $O(\frac{k^2}{\epsilon^2})$ integral points of $\mathcal{P}$ (\cite{mirrokni2017tight}, Proposition 5.1). If we run this algorithm with precision level $\frac{\epsilon}{b}$, Lemma \ref{lemma:lipschitz} guarantees that $|F_i(x_{\text{int}}) - F_i(x)| \leq \epsilon$ for all $i$ and hence $F_i(x_{\text{int}}|S_1) \geq \frac{k_1}{k}\left(1 - \frac{1}{e}\right)\left(W_i - f_i(S_1)\right) - 2\epsilon$. Applying Lemma 2 of \cite{udwani2018multi} (who summarize the guarantee for swap rounding proved by \cite{chekuri2010dependent}), carrying out $O\left(\log \frac{1}{\delta}\right)$ iterations of swap rounding and taking the best outcome produces a set $S_2$ which satisfies $f(S_2|S_1) \geq (1 - \epsilon)\frac{k_1}{k}\left(1 - \frac{1}{e}\right)\left(W_i - f_i(S_1)\right) - 3\epsilon$ with probability at least $1-\delta$, provided that the best outcome is determined by calling a value oracle with precision level $\epsilon$. Adding up the final guarantee, we have 
		
		\begin{align*}
		f(S) &= f(S_1 \cup S_2) \\
		&= f(S_1) + f(S_2|S_1) \\
		&\geq (1 - \epsilon)\frac{k_1}{k}\left(1 - \frac{1}{e}\right)W_1 - 3\epsilon\\
		&\geq (1 - \epsilon)\left(1 - \frac{m}{k\epsilon^3 (1+ \epsilon')}\right)\left(1 - \frac{1}{e}\right)W_1 - 2\epsilon
		\end{align*}
		
		and now rescaling $\epsilon$ by a factor $\frac{1}{3}$ gives the final approximation guarantee. The asymptotic $1-1/e$ approximation follows by setting $\epsilon$ as in \cite{udwani2018multi}.  
		
		We now add up the final runtime. The first thresholding step requires $n$ value oracle calls to each of the $m$ objectives at precision level $\epsilon'$. \textsc{MultiobjectiveFW} requires $\frac{bk^2}{\epsilon}$ iterations, each of which calls mirror descent once. Each invocation of mirror descent requires a total of $O\left(\frac{1}{\epsilon^4} \left(c_{\text{grad}}\sqrt{k \log n} + c_{\text{item}} k \sqrt{\log n}\right)^2 \log \frac{bk}{\delta\epsilon}\right)$ oracle calls. Recalling that $c = \max\{c_{\text{item}}, c_{\text{grad}}\}$, this is upper bounded by $O\left(\frac{c^2 k^2}{\epsilon^4}  \log \left(n + \frac{bk}{\delta\epsilon}\right)\right)$. Each iteration of \textsc{MultiobjectiveFW} also uses  $m$ value oracle calls at precision level $\epsilon$. Finally, each iteration uses additional $O(n+m)$ overhead, for a total of $O\left(\frac{(n + m) k^2 b}{\epsilon}\right)$. In the rounding procedure, we first need to involve \textsc{ApproximateCaratheodory} with precision level $\frac{\epsilon}{b}$, which per Proposition 5.1 of \cite{mirrokni2017tight} requires $\frac{k^2b^2}{\epsilon}$ iterations, and one linear maximization over $\mathcal{P}$ per iteration. Since $\mathcal{P}$ is the uniform matroid polytope, each linear maximization takes time $O(n)$, and so this stage contributes time $O\left(\frac{nk^2b^2}{\epsilon}\right)$. Lastly, we have the $O\left(\log \frac{1}{\delta}\right)$ iterations of swap rounding. Since $x_{\text{int}}$ was decomposed into $\frac{k^2b^2}{\epsilon^2}$ integral points, swap rounding takes time $\frac{k^3b^2}{\epsilon^2}$ for each iteration \cite{chekuri2010dependent}. We also need one $\epsilon$-accurate value oracle call to each of the objective functions per iteration so that we can select the (approximately) best set. Combining these bounds results in the final stated runtime. 
	\end{proof}

	\section{Efficient stochastic gradient estimates}
	
	We now give efficient implementations for the oracles $\oracle{}{grad}$ and $\oracle{}{item}$. They run in combined time $O\left(k\left(|V| + |E|\right) \log^2\frac{|V|}{\delta}\right)$ time, where the operation succeeds with probability $1-\delta$. Our implementations guarantee $c \leq 2b$ whenever they succeed. 
	
	We use a representation the influence maximization objective as the expectation over a set of deterministic submodular functions. Specifically, we can view the independent cascade model as specifying a distribution over \emph{live-edge} graphs \cite{kempe2003maximizing} where each edge is present with probability $p$ and absent otherwise (where all events are independent). Let $\xi$ denote a graph realized from this process, which we will denote $\xi \sim P$. For a fixed $\xi$, the influence spread of a given seed set $S$ is just the number of nodes which are reachable from $S$ via only the edges present in $\xi$. We will denote this quantity by $f(S, \xi)$, where $f(S) = \E_{\xi \sim P}[f(S, \xi)]$.  
	
	
	The starting point is to recall that for any group's utility function $f_i$, the gradients of the multilinear extension $F_i$ satisfy
	
	\begin{align}
	\nabla_{x_j} F_i &= \E_{S \sim x}[f_i(S \cup \{j\}) - f_i(S \setminus \{j\})]\nonumber\\
	&= \E_{S \sim x, \xi \sim P}[f_i(S \cup \{j\}, \xi) - f_i(S \setminus \{j\}, \xi)],\label{eq:gradient-diff}
	\end{align}
	
	which follows from the definition of the multilinear extension \cite{chekuri2010dependent}. Note that for any fixed $i$ and $x_j$, we can obtain a stochastic estimate of this quantity in time $O(|V| + |E|)$ by first drawing a set $S \sim x$, simulating the cascade process, and counting the number of of nodes reached with and without item $j$. By submodularity, the resulting estimate satisfies $f(S \cup \{j\}, \xi) - f(S \setminus \{j\}, \xi) \leq b$ for any $S$ and $\xi$. Naively repeating this process over all $i,j$ would hence require time $O(|V|(|V| + |E|)m)$. We now show how to implement the required oracles by drawing a number of samples that scales only with $k \log |V|$ instead of $|V|$. 
	
	Implementing $\oracle{}{item}$ is simpler because we only need to estimate  $\left[\nabla_{x_j} F_1(x) ... \nabla_{x_j} F_m(x)\right]$ for a single fixed $x_j$. Hence, we can draw a single $S, \xi$, count the number of nodes reachable in each group under $\xi$ with set $S \setminus \{j\}$, and then count the number of nodes reachable with set $S \cup \{j\}$. This takes time $O\left(|V| + |E|\right)$. 
	
	Efficiently implementing $\oracle{}{grad}$ is more difficult since we need to simultaneously estimate $\nabla F_i$ with respect to every $x_j$; hence, naive enumeration would take $O(|V|^2)$ time. We now detail our strategy.	We start by considering a given sample $(S, \xi)$ and show how to estimate the marginal contribution $f_i(S \cup \{j\}, \xi) - f_i(S, \xi)$ for a given $i$ and and \emph{all} $j \not\in S$ in total runtime $O\left(\left(|V| + |E|\right) \log\frac{|V|}{\delta}\right)$. We first remove all nodes from $G$ that are reachable from $S$ under $\xi$, which takes time $O\left(|V| + |E|\right)$. Any node removed in this stage has marginal contribution 0. Next, we remove all nodes that are isolated in the remaining subgraph and assign them marginal contribution 1 if they are part of group $i$. This stage takes time $O(|V|)$.

	
	
	Now we deal with the remaining nodes. Here, determining their marginal contribution of node $v$ to group $i$ amounts to estimating the number of nodes of group $i$ which are reachable from $v$ in $\xi$. We use the size estimation framework of Cohen \cite{cohen1997size}, which allows us to simultaneously produce an unbiased estimate of every remaining node's contribution to group $i$ in time $O\left(|E|\right)$. We apply the weighted version of the estimator, where every node in group $i$ has weight 1 and all other nodes have weight 0. We take  $O(\left(\log \frac{|V|}{\delta}\right)$ independent repetitions of the estimation process, resulting in $O\left(|E|\log\frac{|V|}{\delta}\right)$ runtime.  For a given group $i$, and using $\ell$ repetitions, Cohen's estimator produces an estimate $\Delta(v)$ for each node which satisfies
	
	\begin{enumerate}
	    \item $\E[\Delta(V)] = f_i(\{v\}|S)$
	    \item $\Pr\left[|\Delta(v) - f_i(\{v\}|S)| \geq \epsilon f_i(\{v\}|S)\right] \leq e^{-\Omega\left(\epsilon^2 \ell\right)}$ for any $0 \leq \epsilon \leq 1$
	\end{enumerate}
	
	We fix $\epsilon = 1$ as an arbitrary constant and use $\ell = O\left(\log \frac{|V|}{\delta}\right)$. This allows us to use union bound combined with the second property of the estimator to argue that over all nodes combined
	
	\begin{align*}
	    \Pr\left[\Delta(v) \geq  2b\right] \leq \Pr\left[\Delta(v) \geq  2f_i(\{v\}|S)\right] \leq \delta
	\end{align*}
	
	and so the resulting gradients will satisfy our stated bounds on $c_{\text{item}}$ and $c_{\text{grad}}$ with high probability.

	
	Our overall strategy is to generate enough samples that every node is missing from $S$ in at least one of them. Then, we can use a node's marginal contribution in the sample from which it missing as its gradient estimate. Note that a node $j$ is absent from any given sample with probability $1 - x_j$. Given budget $k$, at most $\frac{k}{1 - \frac{1}{k+1}} = k+1$ nodes can have $x_j \geq 1 - \frac{1}{k+1}$. For any such node, we can explicitly estimate a sample of Equation \ref{eq:gradient-diff} using $O\left(|V| + |E|\right)$ time per node, for $O\left(k\left(|V| + |E|\right)\right)$ total. For the remaining nodes, a simple argument shows that taking $(k+1) \log\frac{|V|}{\delta}$ samples is sufficient to ensure that each node is missing from at least one sample with combined probability $1-\delta$. Summing up, the total runtime to implement $\oracle{}{grad}$ is $O\left(k\left(|V| + |E|\right) \log^2\frac{|V|}{\delta}\right)$. 
	
	
	\section{Runtime comparison with previous work}
	
	The best previous algorithm for multiobjective submodular maximization \cite{udwani2018multi} uses the same overall framework as us, but uses a MWU algorithm for the second stage (the continuous maximization problem). The MWU algorithm runs $O\left(\frac{m}{\epsilon^2}\right)$ iterations, where each iteration requires a call to a greedy algorithm that maximizes a weighted combination of the $f_i$. Using the best implementation of the greedy algorithm \cite{badanidiyuru2014fast}\footnote{While there are efficient special-purpose techniques for influence maximization on a given graph, it is not obvious how to adapt them to deal with the weighted combination of group objectives.}  requires $O\left(\frac{n}{\epsilon}\log \frac{n}{\epsilon}\right)$ value oracle calls, for $O\left(\frac{n}{\epsilon^3} \log m \log \frac{n}{\epsilon}\right)$ such calls in total. By comparison, our algorithm accesses the function through calls to the gradient oracles $\oracle{}{item}$ and $\oracle{}{grad}$. It makes a number of calls to these oracles which is only logarithmic in $n$, scaling as $O\left(\frac{bc^2k^4}{\epsilon^3}  \log \left(n + \frac{bk}{\delta\epsilon}\right)\right)$. Since gradient oracle calls can typically be implemented in similar asymptotic runtime to value oracle calls for common classes of functions (as we have demonstrated for influence maximization), our algorithm effectively saves a factor $O(n)$ runtime in exchange for worse dependence on $k$ and $b$. Since we expect $n$ to grow much faster than $k$ or $b$ (in many typical applications, $b$ is a small constant \cite{hassani2017gradient}), this is often an improvement in asymptotic runtime. For influence maximization in particular, it is easy to see that a value oracle call for a given group cannot be implemented in less than $O(|V| + |E|)$ time, which matches (up to log factors) our stochastic gradient oracle's dependence on the graph size.

\end{document}